\documentclass{article}
\usepackage{graphics,epsfig}
\usepackage{amssymb}
\usepackage{graphicx}
\usepackage{amsmath}
\usepackage{amscd}
\usepackage{mathrsfs}
\usepackage{latexsym}
\usepackage{amsthm}
\usepackage{bbm}
\usepackage{graphicx}
\usepackage{sectsty}
\usepackage{appendix}
\usepackage{float}
\usepackage{url}
\usepackage{caption}

\usepackage{amsthm}

\title{A Quantum Probability Explanation in Fock Space for Borderline Contradictions}
\author{Sandro Sozzo\footnote{ss831@le.ac.uk, ssozzo@vub.ac.be} \\
School of Management, University of Leicester, Leicester, United Kingdom \\
Center Leo Apostel for Interdisciplinary Studies (CLEA) \\ 
Vrije Universiteit Brussel (VUB), Brussels, Belgium}
\begin{document}
\maketitle
%\begin{frontmatter}



%\author[da]{Diederik Aerts}
%\ead{diraerts@vub.ac.be}
%\ead[url]{http://www.vub.ac.be/CLEA/people/aerts}
%\address[da]{Center Leo Apostel for Interdisciplinary Studies (CLEA), Vrije Universiteit Brussel (VUB), Krijgskundestraat 33, 1160 Brussels, Belgium \\Departments of Mathematics and Psychology, Vrije Universiteit Brussel (VUB), Pleinlaan 2, 1160 Brussels, Belgium}

%\address{}

\begin{abstract}
The construction of a consistent theory for structuring and representing how concepts combine and interact is one of the main challenges for the scholars involved in cognitive studies. All traditional approaches are still facing serious hindrances when dealing with `combinations of concepts' and `conceptual vagueness'. One of the main consequences of these difficulties is the existence of `borderline cases' which is hardly explainable from the point of view of classical (fuzzy set) logic and probability theory. Resting on a quantum-theoretic approach which successfully models conjunctions and disjuncions of two concepts, we propound a quantum probability model in Fock space which accords with the experimental data collected by Alxatib and Pelletier (2011) on borderline contradictions. Our model allows one to explain the occurrence of the latter contradictions in terms of genuine quantum effects, such as `contextuality', `superposition', `interference' and `emergence'. In particular, we claim that it is the specific mechanism of `emergence of a new concept' that is responsible of these deviations from classical logical thinking in the cognitive studies on human thought. This result seems to be compatible with a recent interesting application of quantum probabilistic modeling in the study of borderline vagueness (Blutner, Pothos \& Bruza, 2012), and analogies and differences with it are sketched here.

\vspace{.2cm}
\end{abstract}

%\end{keyword}

%\end{frontmatter}

% \linenumbers


\section{Introduction\label{intro}}
Shedding light on the mechanism and dynamics of concept combination would enhance fundamental aspects of a deeper understanding of human thought. It would give us new insight in how sentences and texts are formed by simple concept combinations and as a consequence on how meaning is carried by conceptual communication between human minds. The identification of new aspects of such a mechanism would also have an impact on various disciplines, for example psychology, linguistics, computer science and artificial intelligence. The state of affairs is however that none of the existing theories on concepts allows to identify a mechanism of `how concepts combine', that is, derive the model that represents the combination of two or more concepts from the models that represent the individual concepts. This `combination problem' was identified in a crucial way by Hampton's experiments (Hampton, 1988a,b) which measured the deviation from classical set theoretic membership weights of exemplars with respect to pairs of concepts and their conjunction or disjunction. Hampton's investigation was motivated by the so-called `Guppy effect' in concept conjunction identified by Osherson and Smith (1981). These authors considered the concepts {\it Pet} and {\it Fish} and their conjunction {\it Pet-Fish}, and observed that, while an exemplar such as {\it Guppy} is a very typical example of {\it Pet-Fish}, it is neither a very typical example of {\it Pet} nor of {\it Fish}. Hence, the typicality of a specific exemplar with respect to the conjunction of concepts shows an unexpected behavior from the point of view of classical set and probability theory. As a result of the work of Osherson and Smith (1981), the problem is often referred to as the `Pet-Fish problem' and the effect is usually called the `Guppy effect'. Hampton identified a Guppy-like effect for the membership weights of exemplars with respect to pairs of concepts and their conjunction (Hampton, 1988a), and equally so for the membership weights of exemplars with respect to pairs of concepts and their disjunction (Hampton, 1988b). Several experiments have since been conducted (Hampton, 2001), and many elements have been taken into consideration with respect to the Pet-Fish problem, to provide a satisfactory mathematical model of concept combinations. In particular, we refer to the fuzzy set based attempts (Osherson \& Smith, 1982; Zadeh, 1965; Zadeh, 1982), concerning the `Guppy effect', and to the explanation based theories (Fodor, 1994; Komatsu, 1992; Rips, 1995), concerning concept combinations. However no mechanism and/or procedure has as yet been identified that gives rise to a satisfactory description or explanation of the effects appearing when concepts combine. The combination problem is considered so serious that sometimes it is stated that not much progress is possible in the field if no light is shed on this problem (Fodor, 1994; Hampton, 1997; Kamp \& Partee, 1995; Rips, 1995).

Directly connected with the problem of conceptual vagueness and graded membership (Osherson \& Smith, 1997) are the so-called `borderline contradictions' (Bonini, Osherson, Viale \& Williamson, 1999; Alxatib \& Pelletier, 2011). Roughly speaking, a borderline contradiction is a sentence of the form $P(x) \land \lnot P(x)$, for a vague predicate $P$ and a borderline case $x$. For example, the sentence ``Mark is rich and Mark is not rich'' constitutes a borderline contradiction. Several studies have been conducted on the possible theories of vague language which can describe such borderline cases (Bonini, Osherson, Viale \& Williamson, 1999; Alxatib \& Pelletier, 2011; Blutner, Pothos \& Bruza, 2012; Ripley, 2011; Sauerland, 2010), but also in this case the obtained results have not been unanimously accepted.

Meanwhile, it has become evident that quantum structures are systematically and successfully applied in a variety of situations in cognitive and social sciences (Aerts, 2009a; Aerts \& Aerts, 1995; Aerts, Broekaert, Gabora \& Sozzo, 2012; Aerts \& Czachor, 2004;  Aerts \& Gabora, 2005a,b; Aerts, Gabora \& Sozzo, 2012;  Bruza, Kitto, McEvoy \& McEvoy, 2008; Bruza, Kitto, McEvoy \& Nelson 2009; Busemeyer \& Bruza, 2012; Busemeyer \& Lambert-Mogiliansky, 2009; Busemeyer, Wang \& Townsend, 2006; Busemeyer, Pothos, Franco \& Trueblood, 2011; Franco, 2009; Khrennikov, 2010; Lambert-Mogilansky, Zamir \& Zwirn, 2009; Melucci, 2008; Pothos \& Busemeyer, 2009; van Rijsbergen, 2004; Wang, Busemeyer, Atmanspacher \& Pothos, 2012; Widdows, 2006). For this reason, a {\it Quantum Interaction approach} was born as an interdisciplinary perspective in which the formalisms of quantum theory were used to model specific situations in disciplines different from the micro-world (Bruza, Lawless, van Rijsbergen \& Sofge, 2007; Bruza, Lawless, van Rijsbergen \& Sofge, 2008;  Bruza, Sofge, Lawless, van Rijsbergen \& Klusch, 2009; Song, Melucci, Frommholz, Zhang, Wang \& Arafat, 2011). In addition, the new emergent field of Quantum Interaction focusing on the application of quantum structures in cognitive disciplines was named {\it Quantum Cognition}. In this paper, we mainly deal with the quantum-theoretic approach elaborated within the Brussels group (Aerts, 2009a,b; Aerts \& Aerts, 1995; Aerts, Aerts \& Gabora, 2009; Aerts, Broekaert, Gabora \& Sozzo, 2012; Aerts, Broekaert, Gabora \& Veloz, 2012; Aerts \& D'Hooghe, 2009; Aerts \& Gabora, 2005a,b; Aerts, Gabora \& Sozzo, 2012; Aerts \& Sozzo, 2011; Gabora \& Aerts, 2002). This approach was inspired by a two decade research on the foundations of quantum theory (Aerts, 1999), the origins of quantum probability (Aerts, 1986) and the identification of typically quantum aspects, such as contextuality, emergence, entanglement, interference, superposition, in macroscopic domains (Aerts \& Aerts, 1995; Aerts, Aerts, Broekaert \& Gabora, 2000; Aerts, Broekaert \& Smets, 1999). A `SCoP formalism' was worked out within the Brussels approach which relies on the interpretation of a concept as an `entity in a specific state changing under the influence of a context' rather than as a `container of instantiations' (Aerts \& Gabora, 2005a,b; Gabora \& Aerts, 2002). This representation of a concept was new with respect to traditional approaches (see, e.g., prototype theory (Rosch, 1973; Rosch, 1977; Rosch, 1983), exemplar theory (Nosofsky, 1988; Nosofsky, 1992) and theory theory (Murphy \& Medin, 1985; Rumelhart \& Norman, 1988)), and allowed the authors to provide a quantum representation of the guppy effect explaining at the same time its occurrence in terms of contextual influence (Aerts \& Gabora, 2005a,b). Successively, the mathematical formalism of quantum theory was employed to model the overextension and underextension of membership weights measured by Hampton (1988a,b). More specifically, the overextension for conjunctions of concepts measured by Hampton (1988a) was described as an effect of quantum interference and superposition (Aerts, 2009a; Aerts, Gabora \& Sozzo, 2012), which also play a primary role in the description of both overextension and underextension for disjunctions of concepts (Hampton, 1988b). Furthermore, a specific conceptual combination experimentally revealed the presence of another genuine quantum effect, namely, entanglement (Aerts, 2009a,b; Aerts, Broekaert, Gabora \& Sozzo, 2012; Aerts, Gabora \& Sozzo, 2012; Aerts \& Sozzo, 2011). Finally, also emergence occurs in conceptual processes, and an explanatory hypothesis was put forward according to which human thought is the superposition of `quantum emergent thought' and `quantum logical thought', and that the quantum modeling approach applied in Fock space enables this approach to general human thought, consisting of a superposition of these two modes, to be modeled.

In this paper, we apply the quantum-based approach for conceptual combinations mentioned above to the study of a concrete problem, namely, `borderline vagueness', also showing how and where our treatment is different from existing ones and is innovative. We firstly summarize our quantum-theoretic model of human thought in Fock space in Section \ref{brussels}. We explain how two processes simultaneously occur during a decision, namely, a logical process and a conceptual process and how both processes can be modeled in a Fock space. Then, we supply in Section \ref{alxatib_pelletier} an analysis of the empirical data collected by Alxatib and Pelletier (2011) on borderline cases of the form ``John is tall and not tall'', stressing the conceptual difficulties connected with a classical logical explanation of these data. Successively, we apply in Section \ref{model} the above quantum-theoretic approach in Fock space to the data in Section \ref{alxatib_pelletier}, hence to the borderline cases of the form ``John is tall and John is not tall''. The latter Fock space modeling is successively extended in Section \ref{lastmodeling} to deal with sentences of the form ``John is neither tall nor not tall''. More specifically, we first show that a single classical probability space cannot model these experimental data, and how a quantum probability model can be used, which entails the relevance of complex numbers. Then, we particularize the quantum probability formula for the conjunction of concepts to the specific vague concepts considered by the authors, and calculate the interference term, together with the weights of `quantum emergent thought' and `quantum logical thought' of the Fock space. Our results accord with the empirical data and allow us to propose the explanatory hypothesis that borderline contradictions are not due to a deviation from the basic logical reasoning but, rather, they are due to the presence of genuine quantum effects, including contextuality, superposition, interference and emergence. Furthermore, it follows from our modeling that emergence is the dominant effect in the dynamics of human thought, a conclusion compatible with the identification of emergence as dominant effect in the Fock space modeling of the conjunction of concepts (Aerts, 2009a; Aerts, Gabora \& Sozzo, 2012). In other words, a borderline contradiction is not really a logical contradiction: it `appears' as a contradiction only if logical thought is considered as the dominant dynamics, as it is usually maintained. On the contrary, if one accepts our theoretical explanation, a borderline contradiction is just one of the manifestations of conceptual emergence, as the dominant dynamics of human thought.

A methodological discussion is presented in Section \ref{methods} where it is argued that the model presented here and, more generally, the quantum-theoretic modeling in Aerts (2009a), have both descriptive and explanatory power. In this perspective, we provide some suggestions on how to falsify them in actual cognitive experiments.  

To conclude we observe that our quantum probability model in Fock space presents various analogies with the quantum probability model recently elaborated by Blutner, Pothos and Bruza (2012) in the study on borderline contradictions. We compare the two models in Section \ref{comparison}, stressing that they come to the same conclusions for what concerns the appearence of quantum aspects as a major cause of the observed deviations from classicality in borderline vagueness. We also observe, however, that the two models are different at a structural and a conceptual level.

\section{Quantum-theoretic modeling for concept combination\label{brussels}}
In this section, we present the basic notions and results on conceptual vagueness and combinations that have been attained within our quantum modeling approach and that will be needed for our purposes. Though many of these results can be obtained by representing vagueness, concepts and combinations in a Hilbert space (see the Appendix), we prefer to work out a more general quantum representation for vagueness in Fock space from the very beginning, since this is the formulation we will refer to when dealing with borderline contradictions.\footnote{We observe that an application of Fock spaces to cognitive phenomena can be found in Beim Graben et al. (2008), where the full Fock space is used in natural language processing.} 

In its simplest mathematical form, a Fock space $\mathcal F$, for the case of combining two entities, which is what we focus on here, consists of two sectors: `sector 1' is a Hilbert space $\cal H$, while `sector 2' is a tensor product Hilbert space $\cal H \otimes \cal H$. As a general consideration, sector 1 mainly allows the modeling of interference connected phenomena, while sector 2 mainly allows the modeling of entanglement connected phenomena.

Let us now consider the membership weights of exemplars of concepts and their conjunctions/disjunctions measured by Hampton (1988a,b). He identified systematic deviations from classical set (fuzzy set) conjunctions/disjunctions, an effect known as `overextension' or `underextension'. 

Let us start from conjunctions. Relying on research on the foundations of quantum theory and quantum probability (Pitowsky, 1989), it can be shown that a large part of Hampton's data cannot be modeled in a classical probability space satisfying the axioms of Kolmogorov (1933), due to the following theorem.
\newtheorem{theorem}{Theorem}
\begin{theorem} \label{th1}
The membership weights $\mu(A), \mu(B)$ and $\mu(A\ {\rm and}\ B)$ of an exemplar $x$ for the concepts $A$, $B$ and $`A \ {\rm and} \ B'$ can be represented in a classical probability model if and only if the following two conditions are satisfied. 
\begin{eqnarray} \label{mindeviation}
\Delta_c=\mu(A\ {\rm and}\ B)-\min(\mu(A),\mu(B)) \le 0 \\ \label{kolmogorovianfactorconjunction}
0 \le k_c=1-\mu(A)-\mu(B)+\mu(A\ {\rm and}\ B)
\end{eqnarray}
where $\Delta_c$ is the \emph{conjunction rule minimum deviation}, and $k_c$ is the \emph{Kolmogorovian conjunction factor}.
\begin{proof} See Aerts, 2009a, theorem 3.
\end{proof}
\end{theorem}
Let us briefly comment on Theorem \ref{th1}. Equation (\ref{mindeviation}) expresses compatibility with the minimum rule for the conjunction of fuzzy set theory and, more generally, with monotonicity of classical Kolmogorovian probability. A situation with $\Delta_c >0$ was called \emph{overextension} by Hampton (1988a). Equation (\ref{kolmogorovianfactorconjunction}) expresses instead compatibility with additivity of classical Kolmogorovian probability. Equations (\ref{mindeviation}) and (\ref{kolmogorovianfactorconjunction}) provide together necessary and sufficient conditions to describe the experimental membership weights $\mu(A), \mu(B)$ and $\mu(A\ {\rm and}\ B)$ in a single Kolmogorovian probability space $(\Omega, \sigma(\Omega), P)$ ($\sigma(\Omega)$ being a $\sigma$-algebra of subsets of $\Omega$, $P$ probability measure on $\Omega$). In this case, indeed, events $E_A,E_B\in \sigma(\Omega)$ exist such that $P(E_A)=\mu(A)$, $P(E_B)=\mu(B)$, and $P(E_A \cap E_B)=\mu(A \ {\rm and} \ B)$. 

Let us now consider a specific example. Hampton estimated the membership weight of {\it Mint} with respect to the concepts {\it Food}, {\it Plant} and their conjunction {\it Food and Plant} finding $\mu_{Mint}(Food)=0.87$, $\mu_{Mint}(Plant)=0.81$, $\mu_{Mint}(Food \ {\rm and}\ Plant)=0.9$. Thus, the exemplar \emph{Mint} presents overextension with respect to the conjunction \emph{Food and Plant} of the concepts \emph{Food} and \emph{Plant}. We have in this case $\Delta_c=0.09\not\le0$, hence no classical probability representation exists for these data, because of Theorem \ref{th1}. 

Let us now come to disjunctions. Also in this case, a large part of Hampton's data (1988b) cannot be modeled in a classical Kolmogorovian probability space, due to the following theorem.
\begin{theorem} \label{th3}
The membership weights $\mu(A), \mu(B)$ and $\mu(A\ {\rm or}\ B)$ of an exemplar $x$ for the concepts $A$, $B$ and $`A \ {\rm or} \ B'$ can be represented in a classical probability model if and only if the following two conditions are satisfied. 
\begin{eqnarray} \label{maxdeviation}
\Delta_d=\max(\mu(A),\mu(B))-\mu(A\ {\rm or}\ B)\le 0 \\ \label{kolmogorovianfactordisjunction}
0 \le k_d=\mu(A)+\mu(B)-\mu(A\ {\rm or}\ B)
\end{eqnarray}
where $\Delta_d$ is the \emph{disjunction maximum rule deviation}, and $k_d$ is the \emph{Kolmogorovian disjunction factor}.
\begin{proof} See (Aerts, 2009a), theorem 6.
\end{proof}
\end{theorem}
Let us briefly comment on Theorem \ref{th3}. Equation (\ref{maxdeviation}) expresses compatibility with the maximum rule for the conjunction of fuzzy set theory and, more generally, with monotonicity of classical Kolmogorovian probability. A situation with $\Delta_d >0$ was called \emph{underextension} by Hampton (1988b). Equation (\ref{kolmogorovianfactordisjunction}) expresses instead compatibility with additivity of classical Kolmogorovian probability. Equations (\ref{maxdeviation}) and (\ref{kolmogorovianfactordisjunction}) provide together necessary and sufficient conditions to describe the experimental membership weights $\mu(A), \mu(B)$ and $\mu(A\ {\rm or}\ B)$ in a single Kolmogorovian probability space $(\Omega, \sigma(\Omega), P)$ ($\sigma(\Omega)$. In this case, indeed, events $E_A,E_B\in \sigma(\Omega)$ exist such that $P(E_A)=\mu(A)$, $P(E_B)=\mu(B)$, and $P(E_A \cup E_B)=\mu(A \ {\rm or} \ B)$. 

Let us again consider a specific example. Hampton estimated the membership weight of {\it Donkey} with respect to the concepts {\it Pet}, {\it Farmyard Animal} and their disjunction {\it Pet or Farmyard Animal} finding $\mu_{Donkey}(Pet)=0.5$, $\mu_{Donkey}(Farmyard \ Animal)=0.9$, $\mu_{Donkey}(Pet \ {\rm or}\ Farmyard \ Animal)=0.7$. Thus, the exemplar \emph{Donkey} presents underextension with respect to the disjunction \emph{Pet or Farmyard Animal} of the concepts \emph{Pet} and \emph{Farmyard Animal}. We have in this case $\Delta_d=0.2\not\le0$, hence no classical probability representation exists for these data, because of Theorem \ref{th3}.

It can be proved that a quantum probability model in Fock space exists for Hampton's data (1988a,b). We introduce the essentials of it in the following by using the mathematical formalism resumed in the Appendix (Aerts, 2009a).

Let us consider a concept $A$ and an exemplar (item) $x$. We represent $A$ and $x$ by the unit vectors $|A\rangle$ and $|x\rangle$, respectively, of a Hilbert space $\cal H$, and describe the decision measurement of a subject estimating whether $x$ is a member of $A$ by means of a dichotomic observable represented by the orthogonal projection operator $M$. The probability $\mu(A)$ that $x$ is chosen as a member of $A$, i.e. the membership weight, is given by the scalar product $\mu(A)=\langle A | M|A\rangle$. Let $A$ and $B$ be two concepts, represented by the unit vectors $|A\rangle$ and $|B\rangle$, respectively. To represent the concept $`A \ \textrm{or} \ B'$ we take the archetypical situation of the quantum double slit experiment, where $|A\rangle$ and $|B\rangle$ represent the states of a quantum particle in which only one slit is open, ${1 \over \sqrt{2}}(|A\rangle+|B\rangle)$ represents the state of the quantum particle in which both slits are open, and $\mu(A\ {\rm or}\ B)$ is the probability that the quantum particle is detected in the spot $x$ of a screen behind the slits. Thus, the concept $`A \ \textrm{or} \ B'$ is represented by the unit vector ${1 \over \sqrt{2}}(|A\rangle+|B\rangle)$, and  $|A\rangle$ and $|B\rangle$ are chosen to be orthogonal, i.e. $\langle A | B \rangle=0$. The membership weights $\mu(A), \mu(B)$ and $\mu(A\ {\rm or}\ B)$ of an exemplar $x$ for the concepts $A$, $B$ and $A \ {\rm or} \ B$ are respectively given by
\begin{eqnarray}
\mu(A)&=&\langle A | M|A\rangle \\
\mu(B)&=&\langle B | M|B\rangle \\
\mu(A \ {\rm or} \ B)&=&{1 \over 2}(\mu(A)+\mu(B))+\Re\langle A|M|B\rangle
\end{eqnarray}
where $\Re\langle A|M|B\rangle$ is the real part of the complex number $\langle A|M|B\rangle$. The term $\Re\langle A|M|B\rangle$ is called `interference term' in the quantum jargon, since it produces a deviation from the average ${1 \over 2}(\mu(A)+\mu(B))$ which would have been observed in the quantum double slit experiment in absence of interference. We can see that, already at this stage, two genuine quantum effects, namely, superposition and interference, occur in the mechanism of combination of the concepts $A$ and $B$. In Aerts (2009a) and Aerts, Gabora \& Sozzo (2012) it has been proved that the model above can be realized in the Hilbert space $\mathbb{C}^{3}$ with the interference term given by
\begin{equation}
Int^{d}_{\theta}(A,B)=\sqrt{1-a(A)}\sqrt{1-b(B)}\cos\theta
\end{equation}
with $\theta$ being the `interference angle'. The quantities $a(A)$ and $b(B)$ are defined as $a(A)=1-\mu(A)$ and $b(B)=1-\mu(B)$ if $\mu(A)+\mu(B)\le 1$, $a(A)=\mu(A)$ and $b(B)=\mu(B)$ if $\mu(A)+\mu(B)>1$.

The quantum-theoretic modeling presented above correctly describes a large part of data in Hampton (1988b), but it cannot cope with quite some cases, in fact most of all the cases that behave more classically than the ones that are easily modeled by quantum interference. The reason is that, if one wants to reproduce Hampton's data within a quantum model which fully exploits the analogy with the quantum double slit experiment, one has to include the situation in which two identical quantum particles are considered, both particles passes through the slits, and the probability that at least one particle is detected in the spot $x$ is calculated. This probability is given by $\mu(A)+\mu(B)-\mu(A)\mu(B)$ (Aerts, 2009a). Quantum field theory in Fock space allows one to complete the model, as follows.

In quantum field theory, a quantum entity is described by a field which consists of superpositions of different configurations of many quantum particles. Thus, the quantum entity is associated with a Fock space $\cal F$ which is the direct sum $\oplus$ of different Hilbert spaces, each Hilbert space describing a defined number of quantum particles. In the simplest case, ${\cal F}={\cal H} \oplus ({\cal H}\otimes {\cal H})$, where $\cal H$ is the Hilbert space of a single quantum particle (sector 1 of $\cal F$) and ${\cal H} \otimes {\cal H}$ is the (tensor product) Hilbert space of two identical quantum particles (sector 2 of $\cal F$). Entanglement connected phenomena can obviously be described in ${\cal H} \otimes {\cal H}$.

Let us come back to our modeling for concept combinations. The normalized superposition ${1 \over \sqrt{2}}(|A\rangle+|B\rangle)$ represents the state of the new emergent concept $`A \ \textrm{or} \ B'$ in sector 1 of the Fock space $\cal F$. In sector 2 of $\cal F$, instead, the state of the concept $`A \ \textrm{or} \ B'$ is represented by the unit (product) vector $|A\rangle\otimes|B\rangle$. To describe the decision measurement in this sector, we first suppose that the subject considers two identical copies of the exemplar $x$, pondering on the membership of the first copy of $x$ with respect to $A$ `and' the membership of the second copy of $x$ with respect to $B$. The probability of getting `yes' in both cases is, by using quantum mechanical rules, $(\langle A|\langle B|)|M \otimes M | (|A\rangle\otimes|B\rangle)$. The probability of getting at least a positive answer is instead $1-(\langle A|\langle B|)|(\mathbbmss{1}-M) \otimes (\mathbbmss{1}- M) | (|A\rangle\otimes|B\rangle)$. Hence, the membership weight of the exemplar $x$ with respect to the concept $`A \ \textrm{or} \ B'$ coincides in sector 2 with the latter probability and can be written as  $1-(\langle A|\langle B|)|(\mathbbmss{1}-M) \otimes (\mathbbmss{1}- M) | (|A\rangle\otimes|B\rangle)=\mu(A)+\mu(B)-\mu(A)\mu(B)=(\langle A|\langle B|)| M \otimes \mathbbmss{1}+\mathbbmss{1}\otimes M - M \otimes M | (|A\rangle\otimes|B\rangle)$.

Coming to the Fock space ${\cal F}={\cal H} \oplus ({\cal H}\otimes {\cal H})$, the global initial state of the concepts is represented by the unit vector
\begin{equation}
|\Psi(A,B)\rangle=m e^{i\lambda}|A\rangle\otimes|B\rangle+ne^{i\nu}{1\over \sqrt{2}}(|A\rangle+|B\rangle)
\end{equation}
where the real numbers $m,n$ are such that $0\le m,n$ and $m^2+n^2=1$. The decision measurement on the membership weight of the exemplar $x$ with respect to the concept $`A \ \textrm{or} \ B'$ is represented by the orthogonal projection operator $M \oplus (M \otimes \mathbbmss{1}+\mathbbmss{1}\otimes M - M \otimes M)$, hence the membership weight of $x$ with respect to $`A \ \textrm{or} \ B'$ is given by
\begin{eqnarray} 
\mu(A \ \textrm{or} \ B)=\langle \Psi(A,B)|M \oplus (M \otimes \mathbbmss{1}+\mathbbmss{1}\otimes M - M \otimes M)|\Psi(A,B)\rangle \nonumber\\
=m^2(\mu(A)+\mu(B)-\mu(A)\mu(B))+n^2({\mu(A)+\mu(B) \over 2}+\Re\langle A|M|B\rangle)\label{OR}
\end{eqnarray}
Let us now consider to the representation for the conjunction $`A \ \textrm{and} \ B'$. Here, the decision measurement for the membership weight of the exemplar $x$ with respect to the concept $`A \ \textrm{and} \ B'$ is represented by the orthogonal projection operator $M \oplus (M \otimes M)$, while the membership weight of $x$ with respect to $`A \ \textrm{and} \ B'$ is given by\footnote{The membership weight $\mu(A \ \textrm{or} \ B)$ could have been calculated from the membership weight $\mu(A \ \textrm{and} \ B)$ by observing that the probability that a subject decides for the membership of the exemplar $x$ with respect to the concept $`A \ \textrm{or} \ B'$ is 1 minus the probability of decision against membership of $x$ with respect to the concept $`A \ \textrm{and} \ B'$.}
\begin{eqnarray} 
\mu(A \ \textrm{and} \ B)=\langle \Psi(A,B)|M \oplus (M \otimes M)|\Psi(A,B)\rangle \nonumber \\
=m^2\mu(A)\mu(B)+n^2({\mu(A)+\mu(B) \over 2}+\Re\langle A|M|B\rangle)\label{AND}
\end{eqnarray}
By comparing Equations (\ref{OR}) and (\ref{AND}), it seems that the formulas for the membership weights for conjunction and disjunction of two concepts differ only for the piece in sector 2 of Fock space, $\mu(A)\mu(B)$ versus $\mu(A)+\mu(B)-\mu(A)\mu(B)$, while the piece in sector 1 remains unchanged, which would be counterintuitive. There is a very subtle aspect involved here and it should be clarified. The piece $({\mu(A)+\mu(B) \over 2}+\Re\langle A|M|B\rangle)$ is only formally identical in Equations (\ref{OR}) and (\ref{AND}). We remind, in this respect, that the mathematical representation of the unit vectors $|A\rangle$ and $|B\rangle$ in a concrete Hilbert space generally depends on the exemplar $x$, on the membership weights $\mu(A)$ and $\mu(B)$, and also on whether $\mu(A \ \textrm{or} \ B)$ or $\mu(A \ \textrm{and} \ B)$ is measured, which results in a different interference angle $\theta$. This means that, for a given exemplar $x$, the mathematical representations of $A$ and $B$, hence the interference term $\Re\langle A|M|B\rangle$ is different for $`A \ \textrm{or} \ B'$ and  $`A \ \textrm{and} \ B'$. This was not explicitly mentioned in Aerts (2009a), because Hampton did not measure the disjunction and the conjunction of two concepts for the same exemplars. But, one realizes at once that the interference angles for similar concepts are very different. For example, the interference angles for \emph{Furniture}, \emph{Household appliances}, \emph{Furniture and Household appliances} are very different from those for \emph{House furnishings}, \emph{Furniture}, \emph{House furnishings or Furniture} (see (Aerts, 2009), p. 318). 

The analysis above provides an intuitve theoretical support to the following two theorems.
\begin{theorem} \label{th2}
The membership weights $\mu(A), \mu(B)$ and $\mu(A\ {\rm and}\ B)$ of an exemplar $x$ for the concepts $A$, $B$ and $`A \ {\rm and} \ B'$ can be represented in a quantum probability model where
\begin{equation} \label{membershipweightinterference}
\mu(A\ {\rm and}\ B)=m^2\mu(A)\mu(B)+n^2({\mu(A)+\mu(B) \over 2}+Int^{c}_{\theta}(A,B))
\end{equation}
where the numbers $m^2$ and $n^2$ are convex coefficients, i.e. $0 \le m^2, n^2 \le 1$, $m^2+n^2=1$, and ${\theta}$ is the \emph{interference angle} with
\begin{equation} \label{interferenceterm}
Int^{c}_{\theta}(A,B)=\sqrt{1-a(A)}\sqrt{1-b(B)}\cos\theta
\end{equation}
where $a(A)$ and $b(B)$ are defined as above.
\begin{proof} See (Aerts, 2009a).
\end{proof}
\end{theorem}
The term $\mu(A)\mu(B)$ is compatible with the product t--norm in (classical set) fuzzy logic (Sauerland, 2010). The term $Int^{c}_{\theta}(A,B)$ is instead the quantum interference term and it is responsible, together with the average ${\mu(A)+\mu(B) \over 2}$, of the deviations from classical expectations. For example, in the case of {\it Mint} with respect to {\it Food}, {\it Plant} and {\it Food and Plant}, we have that Theorem \ref{th2} is satisfied with $m_{Mint}^2=0.3$, $n_{Mint}^2=0.7$ and $\theta_{Mint}=50.21^{\circ}$. The numbers $m^2$ and $n^2$ estimate the weight of sectors 2 and 1, respectively, in the Fock space, that is, how the quantum logical thought is correlated with the quantum conceptual thought. This point is crucial and will be extensively discussed in Sections \ref{model} and \ref{lastmodeling} with reference to borderline cases.
\begin{theorem} \label{th4}
The membership weights $\mu(A), \mu(B)$ and $\mu(A\ {\rm or}\ B)$ of an exemplar $x$ for the concepts $A$, $B$ and $A \ {\rm or} \ B$ can be represented in a quantum probability model where
\begin{equation} \label{muAorB}
\mu(A\ {\rm or}\ B)=m^2(\mu(A)+\mu(B)-\mu(A)\mu(B))+n^2( {\mu(A)+\mu(B) \over 2}+Int^{d}_{\theta}(A,B))
\end{equation}
where $m^2$, $n^2$ and ${\theta}$ are defined as in Theorem \ref{th2}.
\begin{proof} See (Aerts, 2009a).
\end{proof}
\end{theorem}
Concerning the {\it Donkey} case, we have that Theorem \ref{th4} is satisfied with $m_{Donkey}^2=0.26$, $n_{Donkey}^2=0.74$ and $\theta_{Donkey}=77.34^{\circ}$. We have explicitly written $Int^{c}_{\theta}(A,B)$ and $Int^{d}_{\theta}(A,B)$ to stress that the two terms are generally different.

Theorems \ref{th2} and \ref{th4} contain the quantum probabilistic expressions allowing the modeling of almost all of Hampton's data (1988a,b). Now, since overextension and underextension cases are so abundant in experimental tests on conjunctions and disjunctions, and since Equations (\ref{membershipweightinterference}), (\ref{interferenceterm}) and (\ref{muAorB}) are so successful in modeling the large collection of data by Hampton (1988a,b), one can naturally wonder about the underlying mechanism and dynamics determining this deviations from classical (fuzzy set) logic and probability theory and, conversely, the effectiveness of a quantum-theoretic modeling. The reason is that a new genuine quantum effect comes into play, namely emergence. Whenever a given subject is asked to estimate whether a given exemplar $x$ belongs to the vague concepts $A$, $B$, $`A \ {\rm and} \ B'$ ($`A \ {\rm or} \ B'$), two mechanisms act simultaneously and in superposition in the subject's thought. A `quantum logical thought', which is a probabilistic version of the classical logical reasoning, where the subject considers two copies of exemplar $x$ and estimates whether the first copy belongs to $A$ and (or) the second copy of $x$ belongs to $B$. But also a `quantum conceptual thought' acts, where the subject estimates whether the exemplar $x$ belongs to the newly emergent concept $`A \ {\rm and} \ B'$ ($`A \ {\rm or} \ B'$). The place whether these superposed processes can be suitably structured is the Fock space. Sector 1 of Fock space hosts the latter process, while sector 2 hosts the former, while the weights $m^2$ and $n^2$ measure the amount of `participation' of sectors 2 and 1, respectively. But, what happens in human thought during a cognitive test is a quantum superposition of both processes. As a consequence of this explanatory hypothesis, an effect, a deviation, or a contradiction, are not failures of classical logical reasoning but, rather, they are a manifestation of the presence of a superposed thought, quantum logical and quantum emergent thought.

Interestingly, the Fock space equation for disjunction can be derived from the Fock space equation for conjunction by applying de Morgan's rules to the `logical reasoning' sector of Fock space. In fact, Equation (\ref{membershipweightinterference}) can be obtained from Equation (\ref{muAorB}) by replacing $\mu(A)\mu(B)$ by $1-((1-\mu(A)(1-\mu(B)))$, which is what we would expect intuitively from the perspective introduced in our modeling. This also means that the de Morgan rules might be not satisfied when both `logical reasoning' and `emergent reasoning' are taken into account, which is potentially relevant for clarifying the situation of borderline vagueness, as we will see in Sections \ref{model} and \ref{lastmodeling}.

The explanation above might seem rather complicate. Nonetheless, we will see in Sections \ref{model} and \ref{lastmodeling} that this explanatory hypothesis accords with experimental data and can be successfully employed to describe the deviations from classical logic and probability theory observed in the borderline contradictions. Let us thus devote the next section to introduce these contradictions.

\section{An experiment measuring borderline vagueness\label{alxatib_pelletier}}
The graded nature of membership weights of exemplars of concepts is a consequence of a characteristic of concepts called `vagueness'. Besides the combination problem discussed in Section \ref{brussels}, the existence of fuzzy boundaries in concepts is responsible of the so-called `Zeno's sorites paradox' and the `borderline contradictions'. We do not dwell with the former in the present paper, for the sake of brevity. We instead introduce the latter in this section.

If we consider the concept {\it A Tall Man}, we can immediately realize that it lacks well defined extensions, since the boundary between `tall' and `not tall' is not clearly established. Moreover, a predicate like `tall man' admits `borderline cases', that is, there are exemplars where it is unclear whether the predicate applies. It is important to observe that this ambiguity cannot be removed by specifying the exact height of the person. It goes without saying that these borderline cases violate some basic rules of classical logic and probability theory, even in their fuzzy set extensions (Bonini, Osherson, Viale \& Williamson, 1999; Alxatib \& Pelletier, 2011; Ripley, 2011). Pioneeristic investigations on these borderline contradictions are due to Bonini et al., (1999), and attempts to solve these difficulties by weakening classical logical rules, proposing pragmatic logics and extensions of probability theory, or resorting to fuzzy set theory, were given by Alxatib and Pelletier (2011), Ripley (2011) and Sauerland (2010). However, despite the promising results that have been obtained, none of the approaches put forward so far provides a faithful modeling of empirical data together with a coherent explanation in terms of concept theory. In the following, we analyze in detail the experiment conducted by Alxatib and Pelletier (2011) on these difficulties.

Alxatib and Pelletier (2011) performed an experiment in which participants were presented with a picture of five suspects of differing heights in a police line-up. The suspects in the line-up were identified by the numbers \#1 ($5'4''$), \#2 ($5'11''$), \#3 ($6'6''$), \#4 ($5'7''$), and \#5 ($6'2''$) and they were shown in a picture purposely not sorted by height, but with an ordering based on names. In addition, participants received a survey with 20 questions and they had to mark one of three check boxes corresponding to three possibilities (true, false, can't tell). The 20 questions consisted of four questions for each of the suspects, as demonstrated below for a given suspect \#x. The ordering of the four questions for each of the five suspects was randomized, so that no two copies of the survey had the same order of questions. A sample of 76 subjects participated in the experiment.

\vspace{2mm}
\begin{tabular}{lccc}
\#x is tall &	True $\Box$	&	False $\Box$ &	Can't tell $\Box$ \\
\#x is not tall &	True $\Box$	&	False $\Box$ &	Can't tell $\Box$ \\
\#x is tall and not tall &	True $\Box$	&	False $\Box$ &	Can't tell $\Box$ \\
\#x is neither tall nor not tall &	True $\Box$	&	False $\Box$ &	Can't tell $\Box$ \\
\end{tabular}

\vspace{2mm}
The results are shown in Table \ref{data}. The figure shows the percentage of true and false answers to the four questions for each subject \#x.
\begin{table}[H] 
\small
\begin{center}
\begin{tabular}{|cccccccccc|}
\hline 
\multicolumn{10}{|l|}{Proposition 1: \it ``Subject \#x is tall''} \\
\multicolumn{10}{|l|}{Proposition 2: \it ``Subject \#x is not tall''} \\
\multicolumn{10}{|l|}{Proposition 3: \it ``Subject \#x is tall and not tall''} \\
\multicolumn{10}{|l|}{Proposition 4: \it ``Subject \#x is neither tall nor not tall''} \\
\hline\hline
\multicolumn{2}{|c}{Subject} & \multicolumn{1}{c}{1: True} & \multicolumn{1}{c}{1: False} &  \multicolumn{1}{c}{2: True} &  \multicolumn{1}{c}{2: False} & \multicolumn{1}{c}{3: True} & \multicolumn{1}{c}{3: False} & \multicolumn{1}{c}{4: True} & \multicolumn{1}{c|}{4: False} \\
\hline
\#1 & {\it 5'4''} & 1.3\% & 98.7\% & 94.7\% & 3.9\% & 14.5\% & 76.3\% & 27.6\% & 65.8\%   \\
\#2 & {\it 5'11''} & 5.3\% & 93.4\% & 78.9\% & 17.1\% & 21.1\% & 65.8\% & 31.6\% & 57.9\%   \\
\#3 & {\it 6'6''} & 46.1\% & 44.7\% & 25.0\% & 67.1\% & 44.7\% & 40.8\% & 53.9\% & 42.1\%   \\
\#4 & {\it 5'7''} & 80.3\% & 10.5\% & 9.2\% & 82.9\% & 28.9\% & 56.6\% & 36.9\% & 55.3\%   \\
\#5 & {\it 6'2''} & 98.7\% & 1.3\% & 0.0\% & 100.0\% & 5.3\% & 81.6\% & 6.6\% & 89.5\%  \\
\hline
\end{tabular}
\end{center}
\caption{Experimental data by Alxatib and Pelletier (2011). \label{data}}
\end{table}
 
A preliminary look at Table \ref{data} reveals that cases exist in which the statement ``$x$ is tall'' is apparently accepted (e.g., subject \#3), in which the statement  ``$x$ is tall'' is apparently rejected (e.g., subjects \#1 and \#4), but also borderline cases in which one gets ``can't tell'' as a typical answer (e.g., subjects \#2 and \#5). One can verify by pure inspection that there is a consistent preference for denying a statement over accepting its negation. Furthermore, there is a substantial preference for rejecting a negation (over accepting a statement). Another relevant result is that there are cases (about 30\%) where the statements ``$x$ is tall'' and ``$x$ is not tall'' are both considered false, whereas the proposition ``$x$ is tall and not tall'' is considered true. The same holds for the statement ``$x$ is neither tall nor not tall''. In addition, accepting the statement ``$x$ is neither tall nor not tall'' is preferred over accepting the statement ``$x$ is tall and not tall''. The latter seems intuitively plausible, but it is difficult to find a theoretical argument for it. And, indeed, Alxatib and Pelletier (2011) could hardly explain the difference, as already noticed in (Blutner, Pothos \& Bruza, 2012). Moreover, there is no evident preference for either rejecting `neither' or rejecting `and'. This is apparent for borderline cases. 

The considerations above on subjects' behavior in the evaluation of borderline cases put at stake some fundamental rules of classical logic, including the de Morgan rules. Alxatib \& Pelletier (2011) provided an exhaustive analysis of their experiment, also in the light of existing concept theories. In particular, they proposed a combination of logic, semantics and pragmatics which provides an intuitive qualitative picture of what is going on at an empirical level. However, their theoretical proposal could not explain the observed pattern for the statement ``$x$ is tall and not tall''. Moreover, they could not explain the fact that the pattern for ``$x$ is tall and not tall'' was different from the pattern for ``$x$ is neither tall nor not tall''. 

A very interesting solution for borderline cases has recently been proposed by Blutner, Pothos \& Bruza (2012) by using quantum probabilistic notions in Hilbert space. These authors point out that a quantum superposition of ``tall'' and ``not tall'' plays a fundamental role in borderline vagueness, and that the quantum interference generated by this superposition is responsible of the deviations from classicality measured by Alxatib and Pelletier (2011). In the next section, we put forward an alternative solution in Fock space which accords with the latter. We will see that the two proposals come to similar conclusions for what concerns the presence of genuinely quantum aspects as the main cause of borderline vagueness. They present instead structural differences in the derivation and use of quantum probability, and also conceptual differences in the role played by superposition and emergence. as well as in the relation between logical and conceptual thinking. A more detailed comparison will be presented in Section \ref{comparison}.

\section{Fock space for borderline contradictions\label{model}}
The Fock space model elaborated in Section \ref{brussels} will be specified in this section to the analysis of a concrete effect and to model the experimental data collected by Alxatib and Pelletier (2011)  within a quantum probabilistic framework. 

Before coming to specific results on the concrete study of borderline vagueness, let we start with a general remark. Our quantum-theoretic modeling in Fock space naturally contemplates the possibility that the classical de Morgan rules are violated in borderline cases. As we have observed in Section \ref{brussels}, in fact, the latter rules hold only in the `logical reasoning' sector of Fock space, but they are generally violated when both `logical reasoning' and `emergent reasoning' are taken into account. Hence, a possible violation of the de Morgan rules in Alxatib \& Pelletier (2011) could be accounted for in our quantum model in Fock space.

Let us now denote by $A$ the concept {\it Tall}, by $A'$ the concept {\it Not Tall} and by $A \ {\rm and} \ A'$ the conjunction {\it Tall and Not Tall}. Then, we denote by $\mu_x(A)$, $\mu_x(A')$ and $\mu_x(A\ {\rm and}\ A')$ the probabilities that a given subject $x$ belongs to the vague concept {\it Tall}, {\it Not Tall} and {\it Tall and Not Tall}, respectively. Then, we assume that, in the large number limits, $\mu_x(A)$, $\mu_x(A')$ and $\mu_x(A\ {\rm and}\ A')$ coincide with the relative frequencies measured by Alxatib and Pelletier (2011) that the propositions ``$x$ is tall'', ``$x$ is not tall'' and ``$x$ is tall and not tall'', respectively, are considered true, for each one $x$ of the exemplars, with $x=1,\ldots,5$. We know that the identification of the `amount of truth' of a given proposition with the `degree of membership' may sound not completely rigorous, though these notions are strongly connected.  Nonetheless, in absence of a concrete experiment measuring typicalities or membership weights, which we plan to perform in the next future, it is reasonable for our purposes to assume that this identification holds.

The above identification allows us to apply the conditions for a classical Kolmogorovian probability model, as well as the quantum probability rules for conjunction in Fock space, to the data collected in Alxatib and Pelletier (2011) and reported in Tables \ref{whole1} and \ref{whole2}.  
\begin{table}[H] 
\small
\begin{center}
\begin{tabular}{|ccccccc|}
\hline 
\multicolumn{2}{|c}{} & \multicolumn{1}{c}{$\mu_{x}(A)$} & \multicolumn{1}{c}{$\mu_{x}(A')$} & \multicolumn{1}{c}{$\mu_{x}(A\ {\rm and}\ A')$} & \multicolumn{1}{c}{$(\Delta_c)_x$} & \multicolumn{1}{c|}{$(k_c)_x$}  \\
\hline
\multicolumn{7}{|l|}{\it $A$=Tall} \\
\multicolumn{7}{|l|}{\it$A'$=Not Tall} \\
\hline
\#1 & {\it 5'4''} & 0.013 & 0.947 & 0.145 & 0.132 & 0.185  \\
\#2 & {\it 5'11''} & 0.461 & 0.250 & 0.447 & 0.197 & 0.736  \\
\#3 & {\it 6'6''} & 0.987 & 0.000 & 0.053 & 0.053 & 0.066  \\
\#4 & {\it 5'7''} & 0.053 & 0.079 & 0.211 & 0.158 & 0.369  \\
\#5 & {\it 6'2''} & 0.803 & 0.092 & 0.289 & 0.197 & 0.394  \\
\hline
\end{tabular}
\end{center}
\caption{Experimental data by Alxatib and Pelletier (2011) for concepts $A$={\it Tall} and $A'$={\it Not Tall}. The probabilities associated with `be tall', `to be not tall', and `to be tall and not tall' are respectively given by $\mu_x(A)$, $\mu_x(A')$ and $\mu_x(A\ {\rm and}\ A')$, for each exemplar $x$. We also reported the classical modeling conjunction factors $(\Delta_c)_x$ and $(k_c)_x$, for each exemplar $x$.\label{whole1}}
\end{table}
Let us start from Table \ref{whole1}. By comparing empirical data with Theorem \ref{th1} in Section \ref{brussels}, we can draw the conclusion that, for each exemplar $x$, a classical probability model does not exist, as stated by the following theorem.
\begin{theorem} \label{noclassical}
The membership weights $\mu_x(A), \mu_x(A')$ and $\mu_x(A\ {\rm and}\ A')$ of an exemplar $x$ for the concepts $A$, $A'$, $A \ {\rm and} \ A'$ cannot be represented in a classical probability model. 
\end{theorem}
\begin{proof} It is sufficient to look at the items in columns $(\Delta_c)_x$ and $(k_c)_x$ of Table \ref{whole1}, for each subject $x$.
\end{proof}
\begin{table}[H] 
\small
\begin{center}
\begin{tabular}{|ccccccccccc|}
\hline 
\multicolumn{2}{|c}{} & \multicolumn{1}{c}{$\mu_{x}(A)$} & \multicolumn{1}{c}{$\mu_{x}(A')$} & \multicolumn{1}{c}{$\mu_{x}(A\ {\rm and}\ A')$} & \multicolumn{1}{c}{$\mu_{x}(A)\cdot\mu_{x}(A')$} & \multicolumn{1}{c}{${1 \over 2}(\mu_{x}(A)+\mu_{x}(A'))$} &  \multicolumn{1}{c}{$\Re\langle A|M|A'\rangle_{x}$} &  \multicolumn{1}{c}{$m^2_x$} & \multicolumn{1}{c}{$n^2_x$} & \multicolumn{1}{c|}{$\theta_x$} \\
\hline
\multicolumn{11}{|l|}{\it $A$=Tall} \\
\multicolumn{11}{|l|}{\it$A'$=Not Tall} \\
\hline
\#1 & {\it 5'4''} & 0.013 & 0.947 & 0.145 & 0.012 & 0.480 & -0.104 & 0.635 & 0.365 & 160.00$^\circ$ \\
\#2 & {\it 5'11''} & 0.461 & 0.250 & 0.447 & 0.115 & 0.356 & 0.198 & 0.243 & 0.757 & 54.36$^\circ$ \\
\#3 & {\it 6'6''} & 0.987 & 0.000 & 0.053 &  0.000 &0.494 & 0 & 0.893 & 0.107 & 0.00$^\circ$ \\
\#4 & {\it 5'7''} & 0.053 & 0.079 & 0.211 &  0.042 & 0.421 & 0.065 & 0.000 & 1.000 & 0.00$^\circ$ \\
\#5 & {\it 6'2''} & 0.803 & 0.092 & 0.289 & 0.074 & 0.448 & -0.073 & 0.283 & 0.717 & 105.67$^\circ$ \\
\hline
\end{tabular}
\end{center}
\caption{Experimental data by Alxatib and Pelletier (2011) for concepts $A$={\it Tall} and $A'$={\it Not Tall}. The probabilities associated with `to be tall', `to be not tall', and `to be tall and not tall' are given by $\mu(A)_x$, $\mu(A')_x$ and $\mu(A\ {\rm and}\ A')_x$, for each exemplar $x$. We also reported the classical expectations $\mu_{x}(A)\cdot\mu_{x}(A')$, the averages ${1 \over 2}(\mu(A)_x+\mu(A')_x)$ and the Fock space weights $m^2_x$ and $n^2_x$, for each exemplar $x$.\label{whole2}}
\end{table}
Let us come to Table \ref{whole2}. By comparing empirical data with Theorem \ref{th2} in Section \ref{brussels}, we can draw the conclusion that, for each exemplar $x$, a quantum probability model exists in Fock space, as stated by the following theorem.
\begin{theorem} \label{yesquantum}
The membership weights $\mu_x(A), \mu_x(A')$, $\mu_x(A\ {\rm and}\ A')$ of an exemplar $x$ for the concepts $A$, $A'$ and $A \ {\rm and} \ A'$ can be represented in a quantum probability model in Fock space. 
\end{theorem}
\begin{proof}
Following the lines of Aerts (2009a), we can construct a quantum model in the Fock space ${\mathbb C}^3\oplus({\mathbb C}^3\otimes{\mathbb C}^3)$ which reproduces the data in Table 3 and such that
\begin{equation} \label{APRIME}
\mu_{x}(A\ {\rm and}\ A')=m_x^2\mu_x(A)\cdot\mu_x(A')+n_x^2({\mu_x(A)+\mu_x(A') \over 2}+\sqrt{\mu_{x}(A)}\sqrt{\mu_{x}(A')}\cos\theta_x)
\end{equation}
for each subject $x$, $x=1, \ldots, 5$. The initial state of the concepts $A, A'$ is represented by the unit vector 
\begin{equation}
|\Psi(A,A')\rangle=m_{x}|A\rangle \otimes |A'\rangle+{n_{x} \over \sqrt{2}}(|A\rangle+|A'\rangle)
\end{equation}
which fits the data in Table 3 with
\begin{eqnarray}
|A\rangle=\Big ( \sqrt{1-\mu_x(A)}, 0, \sqrt{\mu_x(A)} \Big ) \\
|A'\rangle=e^{i \theta_{x}}\Big ( \sqrt{\frac{\mu_x(A)\mu_x(A')}{1-\mu_x(A)}}, \sqrt{\frac{1-\mu_x(A)-\mu_x(A')}{1-\mu_x(A)}}, -\sqrt{\mu_x(A')} \Big ) \\
\theta_{x}= \arccos \Big ( \frac{
{2 \over n_x^2}\Big ( \mu_{x}(A \ \textrm{and} \ A')-m_x^2\mu_x(A)\mu_x(A') \Big )   -\mu_{x}(A)-\mu_{x}(A')}{\sqrt{1-\mu_x(A)}\sqrt{1-\mu_x(A')}} \Big )
\end{eqnarray}
The decision measurement is represented by the orthogonal projection operator $M \oplus (M \otimes M)$, where $M=|100\rangle\langle 100|+|010\rangle\langle 010|$ and $\{|100\rangle, |010\rangle, |001\rangle \}$ is the canonical basis of ${\mathbb C}^{3}$.
\end{proof}
For each subject $x$, the values of $m_x$, $n_x$ and $\theta_x$ fitting the data are reported in Table 3. We note that the Fock space model above contains three parameters $m_x$, $n_x$ and $\theta_x$ which fit the data $\mu_x(A), \mu_x(A')$ and $\mu_x(A\ {\rm and}\ A')$ and are bound by Equation (\ref{APRIME}) and by the condition $m_x^2+n_x^2=1$. Summarizing, there are two independent parameters which are used to predict the data point $\mu_{x}(A \ {\rm and} \ A')$, while $\mu_x(A)$ and $\mu_x(A')$ are inserted into the model and not predicted. We also notice that the unit vectors $|A\rangle$ and $|A'\rangle$ representing the concepts {\it Tall} and {\it Not Tall}, respectively, are orthogonal, which is compatible with quantum mechanical rules, since `tall' and `not tall' can be regarded as outcomes of a decision measurement. 

Let us analyze the results obtained in Table 3. To test the quality of the fit, we have computed the `residual sum of squares' (RSS), getting $\textrm{RSS}=\sum_{x=1}^{5} \Big ( \mu_x(A\ {\rm and}\ A')-m_x^2\mu_x(A)\cdot\mu_x(A')-n_x^2({\mu_x(A)+\mu_x(A') \over 2}+\sqrt{\mu_x(A)}\sqrt{\mu_x(A')}\cos\theta_x)\Big )^{2}=2.26\cdot 10^{-26}$. Then, we note that, for each subject $x$, the average ${1\over 2}(\mu_x(A)+\mu_x(A'))$ and the interference term $\sqrt{\mu_x(A)}\sqrt{\mu_x(A')}\cos\theta_x$ are, in modulus, higher than the corresponding classical probability term 
$\mu_x(A)\cdot\mu_x(A')$. As a consequence, we expect that the effect of emergence is generally prevailing over the logical reasoning. To prove this, we have computed, for each subject $x$, the sum ${1\over 2}(\mu_x(A)+\mu_x(A'))+\sqrt{\mu_x(A)}\sqrt{\mu_x(A')}\cos\theta_x$, thus finding 0.251, 0.519, 0.380, 0.287, 0.395, for the five subjects. After showing that this set has the same variance as the set under $\mu_x(A)\cdot\mu_x(A')$ in Table 3, by means of an F-test, we have then performed an independent 2-sample t-test for these sets of data, finally getting a p-value equal to 0.00013 (1-tail) and 0.000259 (2-tail). This shows that the corresponding data are significantly different from each other. In any case, there are subjects for which the weight of sector 1 is much higher than the weight of sector 2, as follows.

Let us consider subject \#2 as an example. Here, we have an average equal to 0.356, an interference contribution equal to 0.198, for an interference angle equal to 54.36$^\circ$, while the second sector weight is 0.243 and the first sector is 0.757. The explanation is simple, in this case. Whenever a person is asked to estimate whether subject \#2 belongs to the vague concept {\it Tall}, the person first wonders whether subject \#2 belongs to the new emergent concept {\it Tall and Not Tall}, and then the person considers two exemplars of subject \#2 and applies the logical rules wondering whether the first copy belongs to {\it Tall} and the second copy belongs to {\it Not Tall}. But, one can see that the `sector 1 weight' prevails over the `sector 2 weight', and that the average factor prevails over the  interference factor. 

We can even be more explicit to illustrate how the quantum logical thought and the quantum conceptual thought are related in this experiment. Consider, again, subject \#2. In sector 2 of Fock space, two identical copies of subject \#2 are taken into account by a given test participant. One is confronted with {\it Tall}, the other with {\it Not Tall}. If both confrontations lead to acknowledgement of membership, the conjunction is satisfied. We can immediately recognize the classical logical calculus here, except that things are probabilistisc, or fuzzy. The dynamics in sector 1 is instead different. In this sector, a test person estimates whether subject \#2 is a good exemplar of the new emergent concept {\it Tall and Not Tall}. By comparing $m_{2}^2$ with $n_{2}^{2}$, we can see that the quantum conceptual thought dominates over the quantum logical thought, in this case. A similar behavior can be noticed for subject \#5, where $n_5^2=0.717$, $m_5^2=0.283$ and both average and interference terms prevail, in modulus, over the classical probability term.

The above treatment can be repeated for the other subjects involved in the experiment. For example, the other way around occurs for subject \#3. In this case, $m_{3}^2=0.893$ prevails over $n_{3}^{2}=0.107$, which means that quantum logical thought dominates over the quantum conceptual thought, in that case. As a consequence, a classical logical reasoning is the dominant dynamics in this specific case. But, the whole decision process happens in a quantum superposition of the two processes in Fock space, as we have anticipated in Section \ref{brussels}.

We would like to recall that it is not surprising that the non-classical (quantum) effects of superposition, interference, emergence occur in sector 1 of Fock space. Indeed, it is the unit vector ${1 \over \sqrt{2}}(|A\rangle+|A'\rangle)$ in this sector which represents then new emergent concept {\it Tall and Not Tall}. In our analysis of the experiment by Alxatib and Pelletier (2011), the deviations from classical expectations in borderline cases can be mainly explained in terms of these effects of superposition and interference between {\it Tall} and {\it Not Tall}, while no effect of entanglement between {\it Tall} and {\it Not Tall} appears, which could be described in sector 2. This means that the concepts {\it Tall} and {\it Not Tall} are represented by the product vector $|A\rangle\otimes|A'\rangle$ in sector 2, and this sector intuitively only models a classical (logical) behavior in this borderline case. But, entanglement systematically occurs in concepts which is generally represented in sector 2, as it has been proved in Aerts \& Sozzo (2011) and Aerts, Gabora \& Sozzo (2012). In this perspective, new experiments on borderline contradictions might also reveal entanglement between a concept and its negation.

To conclude this section, the quantum model in Fock space presented in this paper for the borderline contradictions identified by Alxatib and Pelletier (2011) shows that, the quantum conceptual dynamics generally prevails over the quantum logical dynamics. This is relevant, from our perspective, since it reveals that these `contradictions' should not be regarded as pure `deviations from classicality' but, rather, as effects due to the dominant dynamics, which is emergence. Hence, the possible violation of de Morgan's rules in borderline cases should be considered as a hint toward the acceptance of emergence as dominant dynamics and classical logical reasoning as secondary dynamics. This is exactly what happens for the deviations from classical (fuzzy set) logic and probability theory in concept combinations discussed in Section \ref{brussels}, and for the `effects' typically experimented in decision theory, such as the {\it conjunction fallacy} and the {\it disjunction effect} (Aerts, 2009a; Aerts \& D'Hooghe, 2009; Busemeyer \& Bruza, 2012; Busemeyer \& Lambert-Mogilansky, 2009; Khrennikov, 2010).

\section{A complete modeling of experimental data\label{lastmodeling}}
The modeling in the previous section does not provide a complete representation of the whole set of data collected by Alxatib \& Pelletier (2011). In particular, the data on the truth value of the setence ``$x$ is neither tall nor not tall'' have not been modeled yet. There is a reason for that, and it has to do that the quantum-theoretic approach presented in Section \ref{brussels} was originally elaborated to cope with simple concept combinations of the form $A\ {\rm and}\ B$ or $A\ {\rm or}\ B$. Hence, it is not trivial that such an approach could model also combinations of the form ${\rm neither} \ A \ {\rm nor} \ B$. We will see in this section that our Fock space modeling in Section \ref{model} is general enough to cope with sentences of the form ``$x$ is neither tall nor not tall'', due to the fact that a concept and its negation are involved in this sentence. But, before undertaking this task, we need to anticipate an important remark, as follows.

When analyzing the semantic content of ``$x$ is neither tall nor not tall'', one must stress that the most natural interpretation of this sentence seems to be ``$x$ is (not tall) and (not (not tall))''. Hence, one should experimentally test the membership weights of a subject $x$ with respect to the concepts {\it Not Tall}, {\it Not (Not Tall)} and {\it Not Tall and Not (Not Tall)}. And, should the latter test be performed, it could reveal that the statistics on ``$x$ is neither tall nor not tall'' is different from the statistics on ``$x$ is (not tall) and (not (not tall))''. Moreover, if we accept such an identification, then one cannot generally use the data on the concept {\it Tall} also for the concept {\it Not (Not Tall)}, since these two concepts are not generally equivalent, as experimentally tested by Hampton (1988b). Notwithstanding this, we remind that sector 1 of Fock space represents a quantum conceptual process, while sector 2 of Fock space represents a quantum logical process in our construction. This entails that it is reasonable to use the same rules that hold in quantum logic, i.e. identifying a quantum proposition with its double negation, when dealing with concepts in this sector. As a consequence, in sector 2, we can treat the conceptual combination {\it Not Tall and Not (Not Tall)} as {\it Not Tall and Tall}, thus using the quantum probabilistic formula for the conjunction in Equation (\ref{APRIME}). More explicitly, we repeat our reasoning in Section \ref{model} and construct a quantum model in the Fock space ${\mathbb C}^3\oplus({\mathbb C}^3\otimes{\mathbb C}^3)$ which reproduces the whole set of data in Table 3, as follows.

In sector 1 of Fock space, the unit vector ${1 \over \sqrt{2}}(|A'\rangle+|A\rangle)$ describes the new emergent concept {\it Neither A Nor A'}, while in sector 2 of Fock space, one should consider the conceptual combination by the tensor product vector $|A'\rangle \otimes |A\rangle$. Thus, the initial state of the concepts $A, A'$ is now represented by the unit vector 
\begin{equation}
|\Psi(A,A')\rangle=m_{x}|A'\rangle \otimes |A\rangle+{n_{x} \over \sqrt{2}}(|A'\rangle+|A\rangle)
\end{equation}
while the decision measurement is represented by the orthogonal projection operator $M \oplus (M \otimes M)$, where $M=|100\rangle\langle 100|+|010\rangle\langle 010|$. This means that we can represent the data on the conceptual combination {\it Neither Tall Nor Not Tall} by the quantum probabilistic formula
\begin{equation} \label{NEITHER}
\mu_{x}({\rm neither} \ A \ {\rm nor} \ A')=m_x^2\mu_x(A)\cdot\mu_x(A')+n_x^2({\mu_x(A)+\mu_x(A') \over 2}+\sqrt{\mu_x(A)}\sqrt{\mu_x(A')}\cos\theta_x)
\end{equation}
for each subject $x$, $x=1, \ldots, 5$. We stress that Equation (\ref{NEITHER}) is only formally equivalent to Equation (\ref{APRIME}), since the values of $m_x^2$, $n_x^2$ and $\theta_x$ are different, as we can see by considering the following table.
\begin{table}[H] 
\small
\begin{center}
\begin{tabular}{|cccccccc|}
\hline 
\multicolumn{2}{|c}{} & \multicolumn{1}{c}{$\mu_{x}(A)$} & \multicolumn{1}{c}{$\mu_{x}(A')$} & \multicolumn{1}{c}{$\mu_{x}({\rm neither}\ A\ {\rm nor}\ A')$} &  \multicolumn{1}{c}{$m^2_x$} & \multicolumn{1}{c}{$n^2_x$} & \multicolumn{1}{c|}{$\theta_x$} \\
\hline
\#1 & {\it 5'4''} & 0.013 & 0.947 & 0.276 & 0.274 & 0.726 & 160.00$^\circ$ \\
\#2 & {\it 5'11''} & 0.461 & 0.250 & 0.539 & 0.170 & 0.830 & 37.28$^\circ$ \\
\#3 & {\it 6'6''} & 0.987 & 0.000 & 0.066 & 0.866 & 0.134 & 0.00$^\circ$ \\
\#4 & {\it 5'7''} & 0.053 & 0.079 & 0.316 & 0.000 & 1.000 & 0.00$^\circ$ \\
\#5 & {\it 6'2''} & 0.803 & 0.092 & 0.369 & 0.241 & 0.759 & 86.79$^\circ$ \\
\hline
\end{tabular}
\end{center}
\caption{Experimental data by Alxatib and Pelletier (2011) for concepts $A$={\it Tall} and $A'$={\it Not Tall}. The probabilities associated with `to be tall', `to be not tall', and `to be neither tall nor not tall' are given by $\mu(A)_x$, $\mu(A')_x$ and $\mu_{x}({\rm neither} \ A \ {\rm nor} \ A')$, for each exemplar $x$.\label{whole3}}
\end{table}
Also in this case, we have computed the residual sum of squares (RSS), getting $\textrm{RSS}=3.61\cdot 10^{-27}$. By comparing Tables 3 and 4, one realizes at once that similar conclusions can be attained with repsect to the prevalence of sector 1 to sector 2 of Fock space, hence to the prevalence of the quantum conceptual thought over the quantum logical thought.

The previous analysis enables to observe that our Fock space modeling can handle situations of the form ``$x$ is neither $A$ nor $B$'', in the specific case when $B$ is the conceptual negation of $A$. We conclude this section by observing that we plan to investigate in the future whether it is possible to further generalize our representation. In this perspective, one should (i) find an emergent superposition state in sector 1 of Fock space for the concept {\it Neither $A$ Nor $B$}, (ii) find a quantum logical expression of {\it Neither $A$ Nor $B$}, which would consist of a suitable combination of $A$ and $B$ in the tensor product Hilbert space. We do not insist on this point, for the sake of brevity.

It must be recalled, to conclude, that the modeling presented in Aerts (2009a) model enables in principle the representation of more complex combinations, exactly because in sector 1 of Fock space `only the emergence of a new concept' plays a role in what happens there. And of course, whatever combination is made of the concepts $A$ and $B$, always a new concept emerges. Hence, also for the concepts {\it Neither $A$ Nor $B$}, {\it $A$ Different From $B$}, {\it $A$ Similar To $B$}, etc. Of course, in the first section of Fock space there is a limitation to what these combinations can be, in the sense that they need to be interpretable within quantum logic. This would not be the case, for example, for {\it $A$ Different From $B$}, or {\it $A$ Similar To $B$}, but it is the case for {\it Neither $A$ Nor $B$}.

\section{Some methodological remarks on the predictions of the Fock space model\label{methods}}
We conclude this paper by providing a brief qualitative and quantitative comparison between the quantum model in Fock space elaborated in this paper and the quantum model in Hilbert space for borderline vagueness developed in Blutner, Pothos \& Bruza (2012). But, first, we think it useful to discuss some general considerations on the methods employed here, which theoretically rests on the results in Aerts (2009a).

One could at first sight observe that the model worked out here, and also the quantum-theoretic framework presented in Aerts (2009a), are descriptive of the corresponding experimental data, but they hardly have explanatory power, due to the fact that they do not supply predictions for future experiments, hence they cannot be falsified. This point is very important, in our opinion, and should be discussed carefully. The model we have presented here has not been devised to directly reproduce the data by Alxatib and Pelletier, thus fitting them. It has rather applied the general perspective on combinations of two concepts in Aerts (2009a) to the case in which the studied combination is the conjunction of a concept and its negation. In this sense, our specialized model can be seen as a confirmation that the quantum-theoretic modeling in Aerts (2009a), which faithfully reproduces Hampton's data (1988a,b), can also be successfully employed to deal with different data sets and in a completely different conceptual framework. The second point is that both models have been constructed to statistically reproduce actual human choices and decisions on the basis of a hypothesized two-layered structure, logical and conceptual, of human thought. This means that our model in Fock space cannot reproduce any arbitrary set of data on borderline cases. For example, one could consider the values for the memebership weights $\mu_x(A)$ and $\mu_x(A')$ collected in Table 2 and the set of values $\mu_{1}(A \ \textrm{and} \ A')=0.658$,
 $\mu_{2}(A \ \textrm{and} \ A')=0.553$, $\mu_{3}(A \ \textrm{and} \ A')=0.579$, $\mu_{4}(A \ \textrm{and} \ A')=0.895$, $\mu_{5}(A \ \textrm{and} \ A')=0.421$ as membership weights of $\mu_{x}(A \ \textrm{and} \ A')$ (we take the real data false answers to the sentence ``$x$ is neither tall nor not tall'' on purpose, using them in place of ``$x$ is tall and not tall''). By following a procedure similar to the one in Section \ref{model}, one can prove that a Fock space model does not exist in this case which fits the new arbitrary data.

We believe that the constraints on our modeling are not imposed by the number of parameters but, rather, by the two-layered Fock space structure assumed in it. In this sense, the constraints we are faced with are similar to the constraints imposed to quantum mechanics by microscopic physics. We agree that further psychological experiments should be performed on cognitive entities to test and eventually falsify these quantum models. For example, they could be put at stake by performing experiments which reveal that a Hilbert space structure is prohibited and cannot work. It could be the case, since we believe that conceptual entities are less crystallized structures than microscopic quantum entities. These new cognitive experiments could, e.g., involve pieces of texts instead of individual exemplars. Concerning the specific Fock space model elaborated in the present paper, it can describe the Alxatib \& Pelletier (2011) data purely in terms of quantum interference and superposition between the concept {\it Tall} and the concept {\it Not Tall}. But the model is also compatible with a situation in which {\it Tall} and {\it Not Tall} also entangle in the standard quantum sense. We expect that superposition, interference and emergence are not the only quantum effects playing a role in borderline vagueness, but also entanglement might play a role here, and even an important one. We plan to investigate this aspect in future work to see whether also entanglement is present in borderline contradictions. But, this will obviously require a suitable choice of the cognitive tests (membership weights on Likert scale, collapse typicality measurements, etc.) to be performed.

Another situation we expect in our quantum model is the following. In principle, for a given subject $x$, the probabilities  $\mu_{x}(A \ \textrm{and} \ A')$ and $\mu_{x}(A' \ \textrm{and} \ A)$ are different in our model. Indeed, although in sector 1 of Fock space the concepts $A \ \textrm{and} \ A'$ and $A' \ \textrm{and} \ A$ are represented by the same superposition state, the phase of this state is different in the two cases, leading to different interference angles, hence to different values for the probabilities. This means that eventual order effects would be automatically accounted for by our model if the combination {\it Not Tall and Tall} is measured instead of {\it  Tall and Not Tall}. These order effects do not seem to occur in Alxatib \& Pelletier (2011),  but they might be revealed by new experiments in which {\it Tall}  and {\it Not Tall} are measured in a different order. 

To conclude this section, we believe that the quantum model in Fock space presented in this paper and, more generally, the quantum theoretic-modeling in Aerts (2009a), do not only describe actual experimental situations, but they also provide an explanation of them in terms of genuine quantum effects.

\section{A comparison between two quantum probability models \label{comparison}}
In this paper, we have elaborated a quantum probability model in Fock space for the data collected by Alxatib and Pelletier. Recently, Blutner, Pothos and Bruza (2012) have worked out an alternative quantum model in Hilbert space for the borderline vagueness occurring in the same experiment. It seems to us interesting to provide a qualitative and quantitative comparison between the two models, analyzing analogies and differences, to see to what extent they are compatible.

In Blutner, Pothos \& Bruza (2012), the authors firstly introduce two dichotomic random variables $\textbf{T}$ and $\textbf{F}$ which can take the values 0,1, instead of working with a single trichotomic ($T,F,N$, N=Null) random variable ``Truth''. The combination $\textbf{T}\textbf{\underline{F}}$ ($T=1,F=0$) corresponds to $T$, and so on, while the combination $\textbf{\underline{T}} \ \textbf{\underline{F}}$ is excluded. Then, they elaborate a classical Kolmogorovian probability model for vagueness, based on the ideas of Alxatib and Pelletier, and show that it cannot quantitatively reproduce original data. Successively, the authors come to a quantum probability model in which the state of ``Tallness'' of a subject $x$ is represented by a unit vector $|\psi_{x}\rangle$ of a Hilbert space, while the decision measurement is represented by a self-adjoint operator whose spectral decomposition contains suitable products of the projection operators $\textbf{T}$, $\textbf{\underline{T}}=\mathbbmss{1}-\textbf{T}$, $\textbf{F}$, $\textbf{\underline{F}}=\mathbbmss{1}-\textbf{F}$. For example, the operator ${1 \over 2}(\textbf{T}\textbf{\underline{F}}\textbf{T}+\textbf{F}\textbf{\underline{T}}\textbf{F})$ is associated with the outcome ``true'' in the decision process, and so on. The crucial assumption in this model is that the projection operators $\textbf{T}$ and $\textbf{F}$ do not commute, which implies that interference terms appear in the probabilities (calculated using the Born rule). This enables faithful modeling for both acceptance and rejection the truth of the propositions ``$x$ is tall'', ``$x$ is not tall'' and ``$x$ is tall and not tall''. The authors write explicitly that this model does not explain the difference between ``$x$ is tall and not tall'' and ``$x$ is neither tall and not tall''. They should reduce to the statement ``$x$ is tall and $x$ is not tall'' versus ``$x$ is not tall and $x$ is tall'' (assuming the law of double negation) (Blutner, Pothos \& Bruza, 2012).

The model briefly summarized above and the one introduced in Sections \ref{model} and \ref{lastmodeling} have some important analogies. We will sketch them in the following, proceeding by steps.

(i) Both approaches accept that a classical set (fuzzy set) modeling cannot cope with borderline vagueness.

(ii) Both models enable to show that a classical Kolmogorovian probability framework cannot satisfactorily describe the data collected by Alxatib and Pelletier.

(iii) Both approaches show that human decisions in this experimental situation can be modeled in a quantum probabilistic framework.

(iv) Both models describe borderline vagueness as an effect of a quantum interference phenomenon between the concepts {\it Tall} and {\it Not Tall}.

There are however some relevant differences between the two models, which can be resumed as follows, again proceeding by steps.

(i) The two proposals seemingly diverge from a technical point of view, since the model presented here is set in a Fock space rather than in a Hilbert space, and each component Hilbert space has a specific role in our treatment. In this respect, the authors wonder whether a Hilbert space modeling exists for Hampton's data too, instead of resorting to a Fock space model, as done in Aerts (2009a). We also observe that the specific representations of states and decision measurements are different in the two models. More specifically, we remind that the concepts {\it Tall} and {\it Not Tall} are represented by orthonormal vectors (equivalently, by orthogonal one-dimensional projection operators) in this paper, while the authors in Blutner, Pothos \& Bruza (2012) assume that the projection operators $\textbf{T}$ and $\textbf{F}$ are neither orthogonal nor commuting, e.g., $\textbf{T}\textbf{\underline{F}}\ne 0$ and $ \textbf{T}\textbf{\underline{F}} \ne \textbf{F}\textbf{\underline{T}}$.

(ii) The presence of a specifically structured two-layered form in human thought, logical and conceptual, is assumed to play a fundamental role in the formation of borderline contradictions in human decisions. This aspect is not present in Blutner, Pothos \& Bruza (2012), which is responsible of the structural difference in the specific quantum models. The distinction of two modes of thought has been proposed in psychological literature. Already Sigmund Freud proposed considering thought as consisting of two processes, which he called `primary' and `secondary' (Freud, 1899). Then, William James introduced the idea of `two legs of thought', `conceptual' and `perceptual' (James, 1910). Jean Piaget, in his study of child thought, introduced `directed or intelligent thought' and `autistic thought' (Piaget, 1923). Finally, Jerome Bruner introduced the `paradigmatic mode of thought' and the `narrative mode thought' (Bruner, 1990) (see also (Sloman, 1996). It is the specific quantum-based structure presented here which distinguishes our distinction from the existing ones.

(iii) In the present paper, we have not considered explicitly the cases of rejection of a statement ``$x$ is tall'', ``$x$ is not tall''. This is because we were not interested in a complete treatment or modeling of the data collected by Alxatib and Pelletier. We were more concerned with showing that the general modeling in Aerts (2009a) can be particularized in such a way that borderline statements can be regarded as a conjunction of a concept and its negation, as the discussion in Section \ref{model} shows. On the contrary, Blutner, Pothos and Bruza (2012) supply a more extensive and complete analysis of Alxatib and Pelletier data, where also the situation of rejection of a statement is quantitatively described. However, we plan to investigate in the future also these aspects which are lacking in the modeling presented here. 

(iv) In Blutner, Pothos \& Bruza (2012), the authors accept that the de Morgan rules are satisfied, which does not affect their quantum modeling. This aspect is more controversial in the model presented in Sections \ref{model} and \ref{lastmodeling}. The two-layered structure of human thought suggests that typical logical relations between concepts should be satisfied only in sector 2 of Fock space, while these logical relations, including de Morgan's rules, could be violated if both sectors are considered. As noticed in (iii), the absence of a complete quantitative treatment of the empirical data in Section \ref{alxatib_pelletier} does not allow us to provide a sharp answer to the question whether de Morgan's laws are violated by our modeling in this specific case. It, moreover, does not allow to conclude that de Morgan's laws can be assumed without affecting the structure of our modeling. 

(v) Order effects seemingly do not play a role in Blutner, Pothos \& Bruza (2012). On the contrary, we expect that order effects might appear if suitable experiments are performed on borderline vagueness, as noticed in Section \ref{methods}. Should this be the case, our modeling can automatically account for the appearence of these order effects.

The analysis in (i)--(v) allows one to conclude that the two quantum models show deep differences but, in absence of a further analysis, which we plan to perform in the next future, the question of the compatibility between them remains an open question.

%\newpage

\appendix
\section*{Appendix. The Basics of Quantum-theoretic Modeling\label{appA}}
When quantum theory is applied for modeling purposes, each entity considered -- in our case a concept -- is associated with a complex Hilbert space ${\cal H}$, which is a vector space over the field ${\mathbb C}$ of complex numbers, equipped with an inner product $\langle \cdot |  \cdot \rangle$, that maps two vectors $\langle A|$ and $|B\rangle$ to a complex number $\langle A|B\rangle$. We denote vectors by using the bra-ket notation introduced by Paul Adrien Dirac, one of the founding fathers of quantum mechanics (Dirac, 1958). Vectors can be kets, denoted by $\left| A \right\rangle $, $\left| B \right\rangle$, or bras, denoted by $\left\langle A \right|$, $\left\langle B \right|$. The inner product between the ket vectors $|A\rangle$ and $|B\rangle$, or the bra-vectors $\langle A|$ and $\langle B|$, is realized by juxtaposing the bra vector $\langle A|$ and the ket vector $|B\rangle$, and $\langle A|B\rangle$ is also called a bra-ket, and it satisfies the following properties: (i) $\langle A |  A \rangle \ge 0$; (ii) $\langle A |  B \rangle=\langle B |  A \rangle^{*}$, where $\langle B |  A \rangle^{*}$ is the complex conjugate of $\langle A |  B \rangle$; $\langle A |(z|B\rangle+t|C\rangle)=z\langle A |  B \rangle+t \langle A |  C \rangle $, for $z, t \in {\mathbb C}$,
where the sum vector $z|B\rangle+t|C\rangle$ is called a `superposition' of vectors $|B\rangle$ and $|C\rangle$ in the quantum jargon. From (ii) and (iii) follows that it is linear in the ket and anti-linear in the bra, i.e. $(z\langle A|+t\langle B|)|C\rangle=z^{*}\langle A | C\rangle+t^{*}\langle B|C \rangle$. We recall that the absolute value of a complex number is defined as the square root of the product of this complex number times its complex conjugate. In formulas, $|z|=\sqrt{z^{*}z}$. Moreover, a complex number $z$ can either be decomposed into its cartesian form $z=x+iy$, or into its goniometric form $z=|z|e^{i\theta}=|z|(\cos\theta+i\sin\theta)$.  As a consequence, we have $|\langle A| B\rangle|=\sqrt{\langle A|B\rangle\langle B|A\rangle}$. We define the `length' of a ket (bra) vector $|A\rangle$ ($\langle A|$) as $|| |A\rangle ||=||\langle A |||=\sqrt{\langle A |A\rangle}$. A vector of unitary length is called a `unit vector'. We say that the ket vectors $|A\rangle$ and $|B\rangle$ are `orthogonal' and write $|A\rangle \perp |B\rangle$ if $\langle A|B\rangle=0$. We have introduced the necessary mathematics to describe the first modeling rule of quantum theory, which is the following. 

\medskip
\noindent{\it First quantum modeling rule:} A state of an entity -- in our case a concept -- modeled by quantum theory is represented by a ket vector $|A\rangle$ with length 1, i.e. $\langle A|A\rangle=1$.

\medskip
\noindent
An orthogonal projection $M$ is a linear function on the Hilbert space, hence $M: {\cal H} \rightarrow {\cal H}, |A\rangle \mapsto M|A\rangle$, which is Hermitian and idempotent, which means that for $|A\rangle, |B\rangle \in {\cal H}$ and $z, t \in {\mathbb C}$ we have (i) $M(z|A\rangle+t|B\rangle)=zM|A\rangle+tM|B\rangle$ (linearity); (ii) $\langle A|M|B\rangle=\langle B|M|A\rangle$ (hermiticity); and (iii) $M \cdot M=M$ (idempotency). The identity, mapping each vector on itself, is a trivial orthogonal projection, denoted by $\mathbbmss{1}$. We say that two orthogonal projections $M_k$ and $M_l$ are orthogonal, if each vector contained in $M_k({\cal H})$ is orthogonal to each vector contained in $M_l({\cal H})$, and we write in this case $M_k \perp M_l$. The orthogonality of the projection operators $M_{k}$ and $M_{l}$ can also be expressed by $M_{k}M_{l}=0$, $0$ being the null operator. A set of orthogonal projection operators $\{M_k\ \vert k=1,\ldots,n\}$ is called a spectral family, if all projectors are mutually orthogonal, i.e. $M_k \perp M_l$ for $k \not= l$, and their sum is the identity, i.e. $\sum_{k=1}^nM_k=\mathbbmss{1}$. This gives us the necessary mathematics to describe the second modeling rule.

\medskip
\noindent
{\it Second quantum modeling rule:} A measurable quantity of an entity -- in our case a concept -- modeled by quantum theory, and having a set of possible real values $\{a_1, \ldots, a_n\}$ is represented by a spectral family $\{M_k\ \vert k=1, \ldots, n\}$ in the following way. If the entity is in a state represented by the vector $|A\rangle$, this state is changed into a state represented by one of the vectors 
\begin{equation}
|A_k\rangle=\frac{M_k|A\rangle}{||M_k|A\rangle||} \nonumber
\end{equation} 
with probability $\langle A|M_k|A\rangle=||M_k |A\rangle||^{2}$. In this case the value of the quantity is $a_k$, and the change of state taking place is called collapse in the quantum jargon. The expression $\langle A|M_k|A\rangle$ is also the probability of getting the outcome $a_k$ in a measurement of the quantity on the entity -- in our case a concept (\emph{Born rule}).

\medskip
\noindent
The tensor product ${\cal H}_{A} \otimes {\cal H}_{B}$ of two Hilbert spaces ${\cal H}_{A}$ and ${\cal H}_{B}$ is the Hilbert space generated by the set $\{|A_i\rangle \otimes |B_j\rangle\}$, where $|A_i\rangle$ and $|B_j\rangle$ are vectors of ${\cal H}_{A}$ and ${\cal H}_{B}$, respectively, which means that a general vector of this tensor product is of the form $\sum_{ij}|A_i\rangle \otimes |B_j\rangle$. This gives us the necessary mathematics to introduce the third modeling rule.

\medskip
\noindent
{\it Third quantum modeling rule:} A state $p$ of a compound entity -- a combined concept -- is represented by a unit vector $|C\rangle$ of the tensor product ${\cal H}_{A} \otimes {\cal H}_{B}$ of the two Hilbert spaces ${\cal H}_{A}$ and ${\cal H}_{B}$ containing the vectors that represent the states of the component entities -- concepts.

\medskip
\noindent
The above means that we have $|C\rangle=\sum_{ij}c_{ij}|A_i\rangle \otimes |B_j\rangle$, where $|A_i\rangle$ and $|B_j\rangle$ are unit vectors of ${\cal H}_{A}$ and ${\cal H}_{B}$, respectively, and $\sum_{i,j}|c_{ij}|^{2}=1$. We say that the state $p$ represented by $|C\rangle$ is a product state if it is of the form $|A\rangle \otimes |B\rangle$ for some $|A\rangle \in {\cal H}_{A}$ and $|B\rangle \in {\cal H}_{B}$. Otherwise, $p$ is called an `entangled state'.

\medskip
\noindent
Fock space is a specific type of Hilbert space, originally introduced in quantum field theory. For most states of a quantum field the number of identical quantum entities is not an actuality, i.e. predictable quantity. Fock space copes with this situation in allowing its vectors to be superpositions of vectors pertaining to sectors for fixed numbers of identical quantum entities. Such a sector, describing a fixed number of $j$ identical quantum entities, is of the form ${\cal H}\otimes \ldots \otimes{\cal H}$ of the tensor product of $j$ identical Hilbert spaces ${\cal H}$. Fock space $F$ itself is the direct sum of all these sectors, hence
\begin{equation} \label{fockspace}
{\cal F}=\oplus_{k=1}^j\otimes_{l=1}^k{\cal H} \nonumber
\end{equation}
For our modeling we have only used Fock space for the `two' and `one quantum entity' case, hence ${\cal F}={\cal H}\oplus({\cal H}\otimes{\cal H})$. This is due to considering only combinations of two concepts. A unit vector $|F\rangle \in {\cal F}$ is then written as $|F\rangle = ne^{i\gamma}|C\rangle+me^{i\delta}(|A\rangle\otimes|B\rangle)$, where $|A\rangle, |B\rangle$ and $|C\rangle$ are unit vectors of ${\cal H}$, and such that $n^2+m^2=1$. For combinations of $j$ concepts, the general form of Fock space expressed in equation (\ref{fockspace}) will have to be used.

\section*{Acknowledgments.} The author is greatly indebted with  Prof. Diederik Aerts for reading the manuscript and providing a number of valuable remarks and suggestions. The author wishes to thank Profs. Reinhard Blutner and Emmanuel Pothos for bringing to the attention of the author the themes connected with borderline vagueness and for deeply discussing the content of the present paper.

\section*{References}
\begin{description}
\vspace{-0.2cm}
\setlength{\itemsep}{-2mm}

%\item Danilov, V.I., Lambert-Mogiliansky, A.: Expected utility theory under non-classical uncertainty. Theory and Decision \textbf{68}, 25--47 (2010).
%\item Aerts, D. (1982a). Description of many physical entities without the paradoxes encountered in quantum mechanics. {\it Foundations of Physics 12}, 1131--1170.
%\item Aerts, D. (1982b). Example of a macroscopical situation that violates Bell inequalities. {\it Lettere al Nuovo Cimento 34}, 107--111.

\item Aerts, D. (1986). A possible explanation for the probabilities of quantum mechanics. {\it Journal of Mathematical Physics 27}, 202--210.

%\item Aerts, D. (1991). A mechanistic classical laboratory situation violating the Bell inequalities with 2$\sqrt{2}$, exactly 'in the same way' as its violations by the EPR experiments. {\it Helvetica Physica Acta 64}, 1--23.

\item Aerts, D. (1999). Foundations of quantum physics: A general realistic and operational approach. {\it International Journal of Theoretical Physics 38}, 289--358.  

\item Aerts, D. (2009a). Quantum structure in cognition. {\it Journal of Mathematical Psychology 53}, 314--348.

\item Aerts, D. (2009b). Quantum particles as conceptual entities: A possible explanatory framework for quantum theory. {\it Foundations of Science 14}, 361--411. 

\item Aerts, D., \& Aerts, S. (1995). Applications of quantum statistics in psychological studies of decision processes. {\it Foundations of Science 1}, 85-97.

\item Aerts, D., Aerts, S., Broekaert, J., \& Gabora, L. (2000). The violation of Bell inequalities in the macroworld. {\it Foundations of Physics 30}, 1387--1414.

%\item Aerts, D., Aerts, S., Coecke, B., D'Hooghe, B., Durt, T. and Valckenborgh, F. (1997). A model with varying fluctuations in the measurement context. In M. Ferrero and A. van der Merwe (Eds.), {\it New Developments on Fundamental Problems in Quantum Physics} (pp. 7-9). Dordrecht: Springer.

\item Aerts, D., Aerts, S., \& Gabora, L. (2009). Experimental evidence for quantum structure in cognition. {\it Lecture Notes in Computer Science 5494}, 59--70.

\item Aerts, D., Broekaert, J. and Smets, S. (1999). A quantum structure description of the liar paradox. {\it International Journal of Theoretical Physics 38}, 3231--3239.

\item Aerts, D., Broekaert, J., Gabora, L., \& Sozzo, S. (2013). Quantum structure and human thought. {Behavioral and Brain Sciences 36}, 274--276.

\item Aerts, D., Broekaert, J., Gabora, L., \& Veloz, T. (2012). The Guppy Effect as Interference. In the {\it Proceedings of the sixth International Symposium on Quantum Interaction}, 27-29 June 2012, Paris, France.

%\item Aerts, D., Coecke, B., \& Smets, S. (1999). On the origin of probabilities in quantum mechanics: creative and contextual aspects. In Cornelis, G., Smets, S., \& Van Bendegem, J. P. (Eds.), {\it Metadebates on Science} (pp. 291-302). Dordrecht: Springer.

\item Aerts, D., \& Czachor, M. (2004). Quantum aspects of semantic analysis and symbolic artificial intelligence. {\it Journal of Physics A: Mathematical and Theoretical 37}, L123--L132. 

\item Aerts, D., \& D'Hooghe, B. (2009). Classical logical versus quantum conceptual thought: Examples in economy, decision theory and concept theory. {\it Quantum Interaction. Lecture Notes in Artificial Intelligence 5494}, 128--142.

%\item Aerts, D. and Durt, T. (1994). Quantum, classical and intermediate, an illustrative example. {\it Foundations of Physics 24}, 1353--1369.

%\item Aerts, D., Durt, T., Grib, A., Van Bogaert, B. and Zapatrin, A. (1993). Quantum structures in macroscopical reality. {\it International Journal of Theoretical Physics 32}, 489--498. 

\item Aerts, D., \& Gabora, L. (2005a). A theory of concepts and their combinations I: The structure of the sets of contexts and properties. {\it Kybernetes 34}, 167--191.

\item Aerts, D., \& Gabora, L. (2005b). A theory of concepts and their combinations II: A Hilbert space representation. {\it Kybernetes 34}, 192--221.

\item Aerts, D., Gabora, L., \& S. Sozzo, S. (2012). Concepts and their dynamics: A quantum--theoretic modeling of human thought. {\it Topics in Cognitive Science} (in print). \emph{ArXiv: 1206.1069 [cs.AI]}.

\item Aerts, D., \& Sozzo, S. (2011). Quantum structure in cognition. Why and how concepts are entangled. {\it Quantum Interaction. Lecture Notes in Computer Science 7052}, 116--127.

%\item Aerts, D., \& Van Bogaert, B. (1992). Mechanistic classical laboratory situation with a quantum logic structure. {\it International Journal of Theoretical Physics 31}, 1839--1848.

\item Alxatib, S., \& Pelletier, J. (2011). On the psychology of truth gaps. In Nouwen, R., van Rooij, R., Sauerland, U., \& Schmitz, H.-C. (Eds.), {\it Vagueness in Communication} (pp. 13--36). Berlin, Heidelberg: Springer-Verlag.

%\item Aristoteles (-350). On Interpretation.

%\item Aspect, A., Grangier, P., \& Roger, G. (1982). Experimental realization of Einstein-Podolsky-Rosen-Bohm gedankenexperiment: A new violation of Bell's inequalities. {\it Physical Review Letters 49}, 91--94.

%\item Bell, J. S. (1964). On the Einstein-Podolsky-Rosen paradox. {\it Physics 1}, 195--200.

\item Beim Graben, P. et al. (2008). Language processing with dynamics fields. {\it Cognitive Neurodynamics 22}, 79--88.

\item Blutner, R., Pothos, E. M., \& Bruza, P. (2012). A quantum probability perspective on borderline vagueness. {\it Topics in Cognition Science} (in print).

\item Bonini, N., Osherson, D., Viale, R., \& Williamson, T. (1999). On the psychology of vague predicates. Mind and Language, 14, 377--393.

\item Bruner, J. (1990). {\it Acts of Meaning}. Cambridge, MA: Harvard University Press.

\item Bruza, P. D., Kitto, K., McEvoy, D., \& McEvoy, C. (2008). Entangling words and meaning. In {\it Proceedings of the Second Quantum Interaction Symposium}. Oxford: Oxford University, 118--124.

\item Bruza, P. D., Kitto, K., Nelson, D., \& McEvoy, C. (2009). Extracting spooky-activation-at-a-distance from considerations of entanglement. {\it Lecture Notes in Computer Science 5494}, 71--83.

\item Bruza, P. D., Lawless, W., van Rijsbergen, C. J., \& Sofge, D., Editors (2007). Proceedings of the AAAI Spring Symposium on Quantum Interaction, March 27--29. Stanford: Stanford University.

\item Bruza, P. D., Lawless, W., van Rijsbergen, C. J., \& Sofge, D., Editors (2008). Quantum Interaction: Proceedings of the Second Quantum Interaction Symposium. London: College Publications.

\item Bruza, P. D., Sofge, D., Lawless, W., van Rijsbergen, K., \& Klusch, M., Editors (2009). Proceedings of the Third Quantum Interaction Symposium. {\it Lecture Notes in Artificial Intelligence 5494}. Berlin: Springer.

\item Busemeyer, J. R., \& Bruza, P. D. (2012). {\it Quantum Models of Cognition and Decision}. Cambridge: Cambridge University Press. 

\item Busemeyer, J. R., \& Lambert-Mogiliansky, A. (2009). An exploration of type indeterminacy in strategic decision-making. {\it Lecture Notes in Computer Science 5494}, 113--127. Berlin: Springer.

\item Busemeyer, J. R., Wang, Z., \ Townsend, J. T. (2006). Quantum dynamics of human decision-making. {\it Journal of Mathematical Psychology 50}, 220--241.

\item Busemeyer, J. R., Pothos, E. M., Franco, R. \& Trueblood, J. S. (2011). A quantum theoretical explanation for probability judgment errors. {\it Psychological Review 118}, 193--218. 

%\item Cardano, G. (1545). Artis Magnae, Sive de Regulis Algebraicis (also known as Ars magna). Nuremberg. 

%\item Clauser, J. F., Horne, M. A., Shimony, A., \& Holt, R. A. (1969). Proposed experiment to test local hidden-variable theories. {\it Physical review Letters 23}, 880--884.

%\item Deerwester, S., Dumais, S. T., Furnas, G. W., Landauer,   T. K. \& Harshman,  R. (1990). Indexing by Latent Semantic Analysis. {\it Journal of the American Society for Information Science 41} 391--407.

\item Dirac, P. A. M. (1958). {\it Quantum mechanics}, 4th ed. London: Oxford University Press.

\item Fodor, J. (1994) Concepts: A potboiler. {\it Cognition 50}, 95--113.

\item Franco, R. (2009). The conjunctive fallacy and interference effects. {\it Journal of Mathematical Psychology 53}, 415--422.

\item Freud, S. (1899). {\it Die Traumdeutung}. Berlin: Fischer-Taschenbuch.

%\item Feynman, R. (1967). {\it The Character of Physical Law}. M.I.T. Press.
%\item Gabora, L., (2001). {\it Cognitive Mechanism Underlying the Origin and Evolution of Culture.} Doctoral Dissertation, Center Leo Apostel, University of Brussels.

\item Gabora, L., \& Aerts, D. (2002). Contextualizing concepts using a mathematical generalization of the quantum formalism. {\it Journal of Experimental and Theoretical Artificial Intelligence 14}, 327--358.

%\item Galea, D., Bruza, P. D., Kitto, K., Nelson, D., McEvoy, C. (2011). Modelling the activation of words in human memory: The spreading activation, Spooky-Activation-at-a-Distance and the entanglement models compared. {\it Lecture Notes in Computer Science 7052}, 149--160.

\item Hampton, J. A. (1988a). Overextension of conjunctive concepts: Evidence for a unitary model for concept typicality and class inclusion. {\it Journal of Experimental Psychology: Learning, Memory, and Cognition 14}, 12--32.

\item Hampton, J. A. (1988b). Disjunction of natural concepts. {\it Memory \& Cognition 16}, 579--591.

\item Hampton, J. A. (1997). Conceptual combination: Conjunction and negation of natural concepts. {\it Memory \& Cognition 25}, 888--909.

\item James, W. (1910). {\it Some Problems of Philosophy}. Cambridge, MA: Harvard University Press.

\item Kamp, H., \& Partee, B. (1995). Prototype theory and compositionality. {\it Cognition 57}, 129--191. 

\item Komatsu, L. K. (1992). Recent views on conceptual structure. {Psychological Bulletin 112}, 500--526 (1992).

%\item Jauch, J. M. (1968). {\it Foundations of Quantum Mechanics}. Addison-Wesley Publishing Company: Reading, Mass.

\item Khrennikov, A. Y. (2010). {\it Ubiquitous Quantum Structure}. Berlin: Springer.

\item Kolmogorov, A. N. (1933). {\it Grundbegriffe der Wahrscheinlichkeitrechnung}, Ergebnisse Der Mathematik; translated as {\it Foundations of Probability}. New York: Chelsea Publishing Company, 1950.

\item Lambert Mogiliansky, A., Zamir, S., \&  Zwirn, H. (2009). Type indeterminacy: A model of the KT(Kahneman-Tversky)-man. {\it Journal of Mathematical Psychology 53}, 349--36.

%\item Laplace, P. S. (1820). Th\'eorie analytique des probabilit\'es. Paris : Mme Ve Courcier.
%\item Li, Y., \& Cunningham, H. (2008). Geometric and quantum methods for information retrieval. {\it ACM SIGIR Forum 42}, 22--32. 
%\item Mackey, G. (1963). {\it Mathematical Foundations of Quantum Mechanics}, Reading: W. A. Benjamin.

\item Melucci, M. (2008). A basis for information retrieval in context. {\it ACM Transactions of Information Systems 26}, 1--41.

\item Murphy, G. L., \& Medin, D. L. (1985). The role of theories in conceptual coherence. {\it Psychological Review 92}, 289-ֳ16.

%\item  Nelson, D. L., McEvoy, C. L. (2007). Entangled associative structures and context. In Bruza, P. D., Lawless, W., van Rijsbergen, C. J., Sofge, D. (Eds.) {\it Proceedings of the AAAI Spring Symposium on Quantum Interaction}. Menlo Park: AAAI Press.

\item Nosofsky, R. (1988). Exemplar-based accounts of relations between classification, recognition, and typicality. {\it Journal of Experimental Psychology: Learning, Memory, and Cognition 14}, 700֭708.

\item Nosofsky, R. (1992). Exemplars, prototypes, and similarity rules. In Healy, A., Kosslyn, S., \& Shiffrin, R. (Eds.), {\it From learning theory to connectionist theory: Essays in honor of William K. Estes}. Hillsdale NJ: Erlbaum.

\item Osherson, D., \& Smith, E. (1981). On the adequacy of prototype theory as a theory of concepts. {\it Cognition 9}, 35--58.

\item  Osherson, D. N., Smith,  E. (1982). Gradedness and Conceptual Combination. {\it Cognition 12}, 299--318.

\item Osherson, D. N., \& Smith, E. (1997). On typicality and vagueness. {\it Cognition 64}, 189--206.

\item Piaget, J. (1923). Le Langage et la Pens\'ee Chez l'Enfant. Paris: Delachaux et Niestl\'e.

\item Pitowsky, I. (1989). {\it Quantum Probability, Quantum Logic}. Lecture Notes in Physics vol. {\bf 321}.  Berlin: Springer.
%\item Piron, C. (1976). {\it Foundations of Quantum Physics}, Reading Mass.: W. A. Benjamin.

%\item Piwowarski, B., Frommholz, I., Lalmas, M., \& van Rijsbergen, K. (2010). What can Quantum Theory bring to IR. In Huang, J., Koudas, N., Jones, G., Wu, X., Collins-Thompson, K., \& An, A. (Eds.), {\it CIKMұ0: Proceedings of the nineteenth ACM conference on Conference on information and knowledge management}.

\item Pothos, E. M., \& Busemeyer, J. R. (2009). A quantum probability explanation for violations of `rational' decision theory. {\it Proceedings of the Royal Society B 276}, 2171--2178.

\item Ripley, D. (2011). Contradictions at the borders. {\it Lecture Notes in Artificial Intelligence}, vol. 6517, 169--188.

\item Rips, L. J. (1995). The current status of research on concept combination. {\it Mind and Language 10}, 72--104.

\item Rosch, E. (1973). Natural categories, {\it Cognitive Psychology 4}, 328--350.

\item Rosch, E. (1978). Principles of categorization. In Rosch, E. \& Lloyd, B. (Eds.), {\it Cognition and categorization}. Hillsdale, NJ: Lawrence Erlbaum, pp. 133--179.

\item Rosch, E. (1983). Prototype classification and logical classification: The two systems. In Scholnick, E. K. (Ed.), {\it New trends in conceptual representation: Challenges to Piagetӳ theory?}. Hillsdale, NJ: Lawrence Erlbaum, pp. 133--159.

\item Rumelhart, D. E., \& Norman, D. A. (1988). Representation in memory. In Atkinson, R. C., Hernsein, R. J., Lindzey, G., \& Duncan, R. L. (Eds.), {\it StevensҠhandbook of experimental psychology}. Hoboken, New Jersey: John Wiley \& Sons.

\item Sauerland, U. (2010). Vagueness in language: The case against fuzzy logic revisited. Unpublished manuscript, Berlin, ZAS.

\item Sloman, S. E. (1996). The empirical case for two systems of reasoning. {\it Psychological Bulletin 119}, 3--22.

\item Song, D., Melucci, M., Frommholz, I., Zhang, P., Wang, L., \& Arafat, S., Editors (2011). Quantum Interaction. {\it Lecture Notes in Computer Science 7052}. Berlin: Springer.

%\item Tversky, A. \& Kahneman, D. (1983). Extension versus intuitive reasoning: The conjunction fallacy in probability judgment. {\it Psychological Review 90}, 293--315.

%\item Tversky, A., \& Shafir, E. (1992). The disjunction effect in choice under uncertainty. {\it Psychological Science 3}, 305--309. 

\item Van Rijsbergen, K. (2004). {\it The Geometry of Information Retrieval}, Cambridge, UK: Cambridge University Press.

%\item  V. J. Weeden, D. L. Rosene, R. Wang, G. Dai, F. Mortazavi, P. Hagmann, J. H. Kaas, and W. I. Tseng (2012). The geometric structure of the brain fiber pathways. {\it Science 335}, 1628--1634.

\item Wang, Z., Busemeyer, J. R., Atmanspacher, H., \& Pothos, E. (2012). The potential of quantum probability for modeling cognitive processes. {\it Topics in Cognitive Science} (in print).

%\item Widdows, D., \& Peters, S. (2003) Word vectors and quantum logic: Experiments with negation and disjunction, in {\it Mathematics of Language 8}, Indiana, IN: Bloomington, pp. 141--154.

\item Widdows, D. (2006). {\it Geometry and Meaning}, CSLI Publications, IL: University of Chicago Press.

%\item Wittgenstein, L. (1953/2001). {\it Philosophical Investigations,  \S 65-71.} Blackwell Publishing.

\item Zadeh, L. (1965). Fuzzy sets. {\it Information \& Control 8}, 338--353 

\item Zadeh, L. (1982). A note on prototype theory and fuzzy sets. {\it Cognition 12}, 291--297.

%\item Zuccon, G., \& Azzopardi, L. (2010). Using the quantum probability ranking principle to rank Interdependent documents. In Gurrin, G., He, Y., Kazai, G., Kruschwitz, U., Little, S., Roelleke, T., R{\"u}ger, S., et al. (Eds.), {\it Advances in Information Retrieval 5993}. Berlin: Springer, pp. 357--369.

\end{description}

\end{document}